\documentclass[a4paper, 11pt]{article}
\usepackage[affil-it]{authblk}
\usepackage{ifpdf}
\usepackage{amsthm,amsfonts,amssymb,amsmath,euscript,comment,mathtools}
\usepackage{tikz}
\usetikzlibrary{decorations.pathmorphing}
\usetikzlibrary{arrows.meta}
\usepackage{mathabx}
\usepackage[utf8]{inputenc}
\usepackage[T1]{fontenc}
\usepackage{graphicx}
\usepackage{color}
\usepackage{textcomp}
\usepackage{bbm}
\usepackage{slashed}
\usepackage[colorlinks=true,linkcolor=red,citecolor=blue]{hyperref}
\usepackage{mathrsfs}
\usepackage{dsfont}
\usepackage{geometry}
\usepackage[all]{xy}
\usepackage[dvipsnames]{pstricks}
\usepackage{epsfig}
\usepackage[nottoc]{tocbibind}
\allowdisplaybreaks[3]
\ifpdf
\renewcommand{\c}{\cdot}
\fi

\usepackage{bm}

\DeclarePairedDelimiter{\norm}{\lVert}{\rVert}

\def\bF{\,^{(\mathbf{F})} \hspace{-2.2pt}\b}
\def\bbF{\,^{(\mathbf{F})} \hspace{-2.2pt}\bb}
\def\rhoF{\,^{(\mathbf{F})} \hspace{-2.2pt}\rho}

\renewcommand{\a}{{\alpha}}

\renewcommand{\b}{{\beta}}

\newcommand{\ga}{\gamma}
\newcommand{\Ga}{\Gamma}
\newcommand{\de}{\delta}
\newcommand{\De}{\Delta}
\newcommand{\ep}{\epsilon}

\newcommand{\la}{\lambda}
\newcommand{\La}{\Lambda}

\newcommand{\Si}{\Sigma}
\newcommand{\om}{\omega}

\newcommand{\vphi}{\varphi}

\renewcommand{\th}{\theta}

\newcommand{\ze}{\zeta}

\newcommand{\nab}{\nabla}
\newcommand{\vsi}{{\varsigma}}
\newcommand{\Up}{\Upsilon}

\newcommand{\ombc}{\widecheck{\omb}}
\newcommand{\Gac}{\widecheck{\Ga}}

\newcommand{\rhoc}{\widecheck{\rho}}
\newcommand{\rhoFc}{\widecheck{\rhoF}}

\newcommand{\Ombc}{\widecheck{\underline{\Omega}}}

\renewcommand{\ombc}{\underline{\widecheck{\omega}}}

\newcommand{\trchc}{\widecheck{\tr\chi}}
\newcommand{\trchbc}{\widecheck{\tr\chib}}

\newcommand{\pr}{{\partial}}
\newcommand{\les}{\lesssim}
\renewcommand{\c}{\cdot}

\newcommand{\trch}{{\tr\chi}}
\newcommand{\chih}{{\widehat \chi}}
\newcommand{\chib}{{\underline \chi}}
\newcommand{\chibh}{{\underline{\chih}}}

\newcommand{\etab}{{\underline \eta}}
\newcommand{\omb}{{\underline{\om}}}
\newcommand{\bb}{{\underline{\b}}}
\renewcommand{\aa}{\protect\underline{\a}}
\newcommand{\xib}{{\underline \xi}}

\newcommand{\ug}{\overset{\circ}{u}}
\newcommand{\sg}{\overset{\circ}{s}}
\newcommand{\rg}{\overset{\circ}{r}}
\newcommand{\mg}{\overset{\circ}{m}}

\newcommand{\ovr}{\rg}

\newcommand{\epg}{\overset{\circ}{\ep}}
\newcommand{\lapzero}{\overset{\circ}{\De}}

\newcommand{\ovg}{\overset{\circ}{g}}

\newcommand{\F}{\mathbf{F}}

\renewcommand{\div}{\sdiv}

\newcommand{\lap}{\De}

\newcommand{\rhod}{\,\dual\hspace{-2pt}\rho}

\newcommand{\Sg}{\overset{\circ}{S}}

\newcommand{\RRR}{\mathcal{R}}
\newcommand{\Mext}{\ext}

\newcommand{\Cb}{\underline{C}}
\newcommand{\Cbp}{\Cb^{(p)}} 
\newcommand{\Mp}{M^{(p)}}
\newcommand{\CbpS}{\Cb^{(\S, p)}} 
\newcommand{\MpS}{M^{(\S, p)}}

\newcommand{\mudot}{\dot{\mu}}
\newcommand{\ext}{{^{(ext)}\MM}}

\newcommand{\ovla}{\overset{\circ}{\la}}

\newcommand{\Gag}{\Ga_g}
\newcommand{\Gab}{\Ga_b}

\newcommand{\D}{\mathbf{D}}

\newcommand{\dkb}{\slashed{\dk}}

\renewcommand{\c}{\cdot}

\DeclareMathOperator{\sdiv}{div}

\newcommand{\fb}{\underline{f}}

\newcommand{\SSS}{{\mathbb{S}}}

\newcommand{\hk}{\mathfrak{h}}

\newcommand{\dk}{\mathfrak{d}}

\newcommand{\ds}{\displaystyle}
\newcommand{\JpS}{{J^{(p, \S)}}} 

\newcommand{\atr}{{}^{(a)}\tr}
\newcommand{\atrch}{\atr\chi}
\newcommand{\atrchb}{\atr\chib}
\newcommand{\muc}{\widecheck{\mu}}

\newcommand{\ov}{\overline}

\newtheorem{thm}{Theorem}[section]

\newtheorem{theorem}[thm]{Theorem}

\newtheorem{remark}[thm]{Remark}
\newtheorem{lemma}[thm]{Lemma}
\newtheorem{proposition}[thm]{Proposition}
\newtheorem{conj}[thm]{Conjecture}

\newtheorem{definition}[thm]{Definition}

\newcommand{\II}{\mathcal{I}}
\newcommand{\MM}{\mathcal{M}}

\DeclareMathOperator{\tr}{tr}

\renewcommand{\th}{\theta}
\renewcommand{\a}{\alpha}
\renewcommand{\b}{\beta}
\newcommand{\g}{{\bf g}}

\newcommand{\ovS}{\Sg}

\renewcommand{\aa}{\underline{\a}}
\renewcommand{\S}{{\mathbf{S}}}

\DeclareMathOperator{\lot}{l.o.t.}
\newcommand{\Jp}{J^{(p)}}
\newcommand{\Mint}{\,{}^{(int)}\mathcal{M}}

\newcommand{\Sitop}{\,^{(top)}\Si}

\DeclareMathOperator{\err}{Err}
\newcommand{\Omb}{\underline{\Omega}}

\numberwithin{equation}{section}

\newcommand{\TT}{{\mathcal T}}

\newcommand{\UU}{{\mathcal U}}

\newcommand{\W}{{\bf W}}

\newcommand{\RR}{{\mathcal R}}

\newcommand{\dual}{{^*}}

 \newcommand{\trchb}{\tr\chib}

\newcommand{\Lab}{\underline{\La}}
\newcommand{\dg}{\overset{\circ}{\delta}}

\newcommand{\dds}{\slashed{d}^*}
\newcommand{\ddd}{\slashed{d}}

\newcommand{\Kc}{\widecheck{K}}

\DeclareMathOperator{\curl}{curl}

\def\ovJ{\overset{\,\,\circ}{J}}

\def\Jpov{\ovJ\,^{(p)}}
\def\ovu{\overset{\circ}{ u}}
\def\ovs{\overset{\circ}{ s}}
\def\ovb{\overset{\circ}{ b}}
\def\ovm{\overset{\circ}{ m}}
\def\ovQ{\overset{\circ}{ Q}}
\def\undB{\underline{B}}

\def\Cbdot{\dot{\Cb}}
\def\Mdot{\dot{M}}
\def\Cbpdot{{\dot{\Cb}^{(p)}}}
\def\Mpdot{{\dot{M}^{(p)}}}
\def\divzero{{\overset{\circ}{ \div}}}

\def\lapzero{{\overset{\circ}{ \lap}}}

\def\SS{{\mathcal S}}

\makeatother
\allowdisplaybreaks

\usepackage[normalem]{ulem}

\usepackage[sortcites,
backend=biber,
date=year,
style=numeric-comp,
doi=false,
isbn=false,
url=false,
maxcitenames=4,
maxbibnames=4,
eprint=true]{biblatex}

\DeclareNameAlias{author}{family-given}
\AtEveryBibitem{
  \clearfield{month}
  \clearfield{url}
  \clearfield{urlyear}
  \clearfield{urlmonth} 
  \ifentrytype{online}{
    \clearfield{year}}
  {
    \clearfield{eprint}}
}
\renewbibmacro*{volume+number+eid}{%
  \printfield{volume}%
  \setunit*{\addnbspace}
  \printfield{number}%
  \setunit{\addcomma\space}%
  \printfield{eid}}
\DeclareFieldFormat[article]{volume}{\textbf{#1}}
\DeclareFieldFormat[article]{number}{\mkbibparens{#1}}
\DeclareFieldFormat*{title}{\textit{#1}}
\DeclareFieldFormat{journaltitle}{#1\isdot}
\renewbibmacro{in:}{}
\makeatother
\allowdisplaybreaks

\hypersetup{colorlinks=true, linkcolor=black, urlcolor=blue, urlbordercolor={0 1 1}}

\author[1]{Allen Juntao Fang\thanks{allen.juntao.fang@uni-muenster.de}}
\author[2]{Elena Giorgi\thanks{elena.giorgi@columbia.edu}}
\author[3]{Jingbo Wan\thanks{jingbo.wan@sorbonne-universite.fr}}

\affil[1]{\small Mathematics M\"unster, Universit\"at M\"unster }
\affil[2]{\small Department of Mathematics, Columbia University}
\affil[3]{\small Laboratoire Jacques-Louis Lions de Sorbonne Universit\'e}

\begin{document}
\title{Mass-Centered GCM Framework in Perturbations of Kerr(-Newman)}

\maketitle

\begin{abstract}
  The nonlinear stability problem for black hole solutions of the Einstein equations critically depends on choosing an appropriate geometric gauge. In the vacuum setting, the use of Generally Covariant Modulated (GCM) spheres and hypersurfaces has played a central role in the proof of stability for slowly rotating Kerr spacetime \cite{klainermanKerrStabilitySmall2023}. In this work, we develop an alternative GCM framework, that we call \textit{mass-centered}, designed to overcome the breakdown of the standard GCM construction in the charged case, where electromagnetic–gravitational coupling destroys the exceptional behavior of the $\ell=1$ mode of the center-of-mass quantity used in the vacuum analysis. This construction is aimed at the nonlinear stability of Reissner–Nordström and Kerr–Newman spacetimes.

  Our approach replaces transport-based control of the center-of-mass quantity with a sphere-wise vanishing condition on a renormalized $\ell=1$ mode, yielding mass-centered GCM hypersurfaces with modified gauge constraints. The resulting elliptic–transport system remains determined once an $\ell=1$ basis is fixed via effective uniformization and provides an alternative construction in vacuum in the uncharged limit. 
\end{abstract}

\tableofcontents

\section{Introduction}

In this paper we develop a modification and extension of the General Covariant Modulated (GCM) framework, a key ingredient in the proof of the nonlinear stability of the Kerr family for $|a|\ll M$ \cite{klainermanKerrStabilitySmall2023, giorgiWaveEquationsEstimates2024}, introduced in polarized symmetry in \cite{klainermanGlobalNonlinearStability2020} and in full generality in \cite{klainermanConstructionGCMSpheres2022, klainermanEffectiveResultsUniformization2022, shenConstructionGCMHypersurfaces2023}. Our version adapts the framework to perturbations of the charged Kerr–Newman black hole, provides an alternative proof in the Kerr case, and is designed to remain applicable to solutions arising from a broad class of matter models.

\subsection{The nonlinear stability of black hole spacetimes}

The nonlinear stability problem for black hole spacetimes asks whether solutions to the Einstein equations that start as small perturbations of a stationary black hole remain globally regular in the exterior region and asymptotically settle down to another member of the same family. More precisely, one seeks to prove that for initial data sufficiently close to a given black hole solution, the maximal globally hyperbolic development admits a complete null infinity, remains close to a suitable stationary metric, and converges to it at late times.

In the uncharged case, corresponding to the vacuum Einstein equations with zero cosmological constant, the relevant stationary solutions are the Schwarzschild and Kerr families. The stability of Schwarzschild was established through a sequence of works building on earlier linear and nonlinear analyses \cite{dafermosLinearStabilitySchwarzschild2019, klainermanGlobalNonlinearStability2020, dafermosNonlinearStabilitySchwarzschild2021}, while the stability of Kerr, which has been the subject of decades of investigation from its mode stability onward, has recently been established for slowly rotating spacetimes \cite{klainermanKerrStabilitySmall2023,giorgiWaveEquationsEstimates2024}.

\begin{theorem}[\cite{klainermanConstructionGCMSpheres2022, klainermanEffectiveResultsUniformization2022, shenConstructionGCMHypersurfaces2023, klainermanKerrStabilitySmall2023, giorgiWaveEquationsEstimates2024}]\label{thm:KS-intro} Vacuum initial data sets, sufficiently close to Kerr initial data for $|a| \ll M$, have a maximal development with a complete future null infinity and with domain of outer communication which approaches (globally) a nearby Kerr solution.
\end{theorem}

The proof of nonlinear stability for slowly rotating Kerr-de Sitter spacetimes with positive cosmological constant was obtained in \cite{hintzGlobalNonLinearStability2018}.

A fundamental difficulty in any nonlinear stability proof is the choice of gauge. The Einstein equations are invariant under diffeomorphism, and any analysis must fix this freedom in a way that is both geometrically canonical and analytically robust. Geometrically, a good gauge choice should single out a preferred foliation and frame that reflect the physical geometry of the spacetime, eliminating residual coordinate ambiguities. Analytically, the gauge must be compatible with the hyperbolic and elliptic structures of the equations, allowing sharp decay estimates to propagate through the hierarchy of quantities. In practice, the success or failure of a stability proof often hinges on whether one can construct such a gauge and control its evolution globally in time. The central idea which resolves this difficulty in \cite{klainermanGlobalNonlinearStability2020, klainermanKerrStabilitySmall2023} is the introduction and construction of GCM spheres and hypersurfaces on which specific geometric quantities take Schwarzschildian values, see Section \ref{sec:GCM-intro} for a review of the GCM framework of \cite{klainermanConstructionGCMSpheres2022, klainermanEffectiveResultsUniformization2022, shenConstructionGCMHypersurfaces2023}.

In the charged case, governed by the Einstein–Maxwell system, the stationary black holes are given by the Reissner–Nordström and Kerr–Newman families. The coupling between gravitational and electromagnetic fields introduces new structural and analytic challenges, both in the linear theory and in the nonlinear problem. Nevertheless, the stability statement is expected to hold in close analogy to the vacuum case: small perturbations of a charged black hole should remain globally regular in the exterior and decay to a nearby member of the same family.

\begin{conj}[Stability of Kerr-Newman conjecture]
  Theorem \ref{thm:KS-intro} respectively holds for electrovacuum initial data close to Kerr-Newman initial data.
\end{conj}

The present work is part of a program aimed at establishing the nonlinear stability for Reissner–Nordström and Kerr–Newman spacetimes. Our contribution here addresses a key obstacle specific to the charged setting—controlling the evolution of the center-of-mass quantity—by introducing a modification of the gauge-fixing GCM procedure of \cite{klainermanConstructionGCMSpheres2022, klainermanEffectiveResultsUniformization2022, shenConstructionGCMHypersurfaces2023} which is adapted to the Einstein–Maxwell equations.

\subsection{The GCM framework \texorpdfstring{in \cite{klainermanConstructionGCMSpheres2022, klainermanEffectiveResultsUniformization2022, shenConstructionGCMHypersurfaces2023}}{of Klainerman, Szeftel, Shen}}\label{sec:GCM-intro}

The proof of nonlinear stability of Kerr in
\cite{klainermanKerrStabilitySmall2023} relies on a continuity
argument, in which the maximal development of perturbed data is
obtained as the limit of a sequence of finite General Covariant
Modulated (GCM) spacetimes, as represented in Figure
\ref{fig:myfigure} below, whose future boundaries consist of the union
$\mathcal{A} \cup \Sitop \cup \Si_*$, where
$\mathcal{A}, \Sitop, \Si_*$ are spacelike. The spacetime $\MM$ also
contains a timelike hypersurface $\TT$ which divides $\MM$ into an
exterior region called $\Mext$ and an interior region called $\Mint$.

\begin{figure}[ht]
  \centering
  \includegraphics[width=0.65\textwidth]{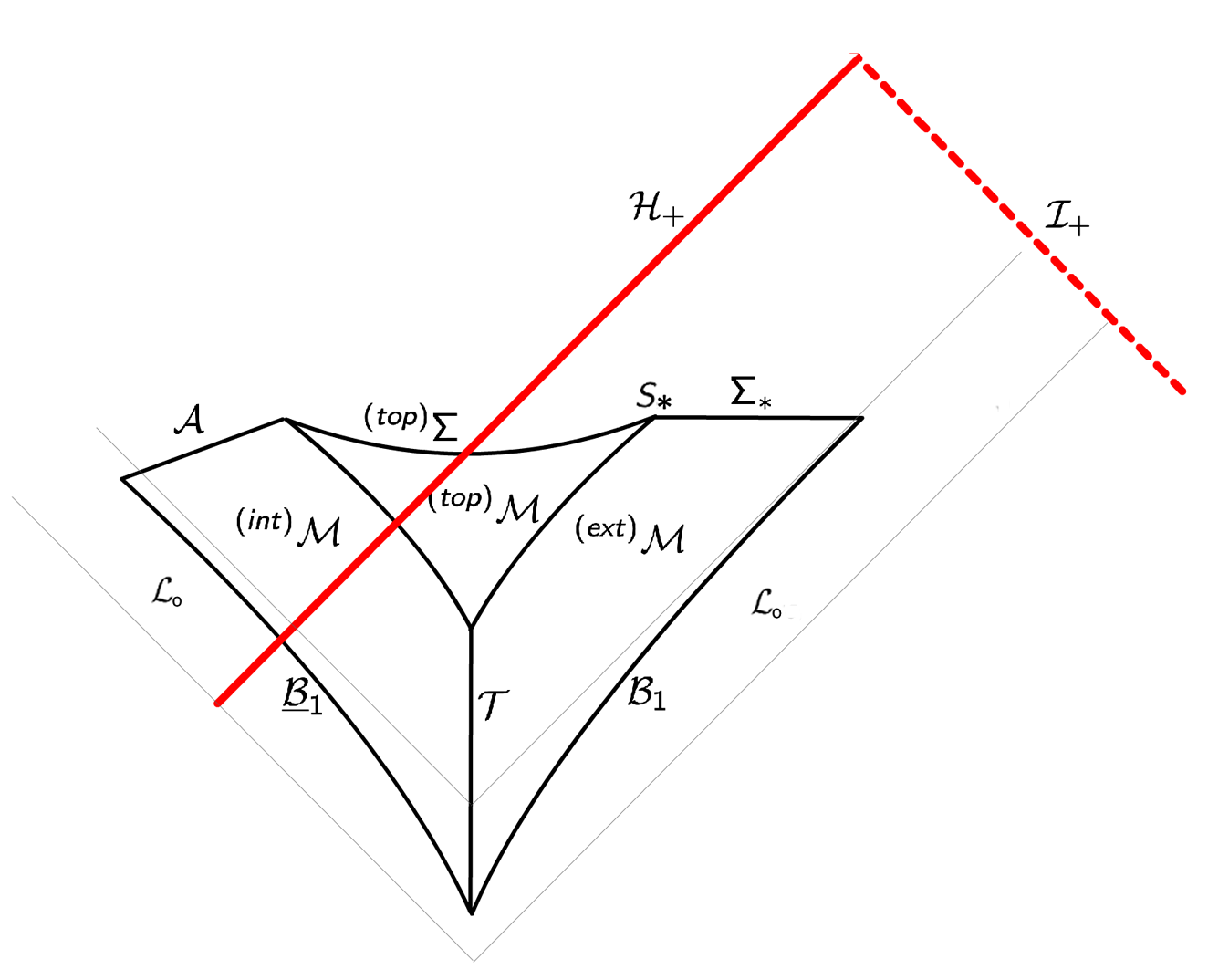} 
  \caption{The GCM admissible spacetime of \cite{klainermanKerrStabilitySmall2023}.}
  \label{fig:myfigure}
\end{figure}

The purpose of the GCM construction is to fix the gauge in the future, thereby removing the large diffeomorphism freedom inherent to the Einstein equations. Within the nonlinear stability program for Kerr spacetime \cite{klainermanKerrStabilitySmall2023,giorgiWaveEquationsEstimates2024}, the GCM framework provides a canonical geometric setting that allows for consistent propagation of geometric quantities and controlled modulation of asymptotic parameters. Since the decay properties of Ricci and curvature components in depend heavily on the choice of the boundary $\Si_*$, the central idea in \cite{klainermanConstructionGCMSpheres2022, klainermanEffectiveResultsUniformization2022, shenConstructionGCMHypersurfaces2023} is the introduction and construction of GCM spheres and hypersurfaces there on which specific values take Schwarzschildian values. In \cite{klainermanKerrStabilitySmall2023}, the final spacetime is constructed as the limit of a continuous
family of finite GCM admissible spacetimes, where at each stage one assumes that all Ricci
and curvature coefficients of a fixed GCM admissible spacetime verify precise
bootstrap assumptions. By making use of the GCM properties of $\Si_*$ and
the smallness of the initial conditions one then shows that all the bounds of the Ricci and
curvature coefficients depend only on the size of the initial data and thus improve the bootstrap assumptions. This allows for to extension of the spacetime to
a larger one in which the bootstrap assumptions are still valid. Note that even though Theorem \ref{thm:KS-intro} holds for slowly rotating Kerr, the GCM framework is independent of the smallness of angular momentum. 

\subsubsection{GCM spheres}

A \emph{GCM sphere} $\S$ is defined as a 2-sphere equipped with an integrable null frame on which the following Schwarzschild-like conditions on the expansions and the mass-aspect functions are imposed:
\begin{align}\label{eq:trch-trchb-mu}
  \trchc = 0, \qquad (\trchbc)_{\ell \geq 2} = 0, \qquad \muc_{\ell \geq 2} = 0,
\end{align}
where the modes $\ell$ are defined in terms of an approximate basis, and here for any quantity $\psi$ we denote $\widecheck{\psi}=\psi - \psi_{Schw}$ the linearization with respect to the corresponding value on Schwarzschild
(see Section \ref{sec:l=1modes} and \ref{sec:def-mass-charge} for precise definitions). 
These constraints eliminate lower-order gauge freedoms and geometrically characterize the sphere relative to the chosen frame. The existence of such GCM spheres was established in \cite{klainermanConstructionGCMSpheres2022}, by deforming given spheres of the background foliation.

To further reduce the residual gauge and select a preferred foliation near null infinity, Klainerman and Szeftel introduced the notion of \emph{intrinsic GCM spheres} \cite{klainermanEffectiveResultsUniformization2022}. These are GCM spheres satisfying in addition
\[
  (\trchbc)_{\ell = 1} = 0, \qquad (\div \b)_{\ell=1} = 0,
\]
with respect to a canonical $\ell=1$ basis determined by an effective uniformization procedure. These extra conditions fix the remaining freedom associated with translations and center-of-mass motion, thereby producing a frame suited for tracking physical quantities at infinity \cite{klainermanCanonicalFoliationNull2024}.

\subsubsection{GCM hypersurfaces}
Building on this, Shen \cite{shenConstructionGCMHypersurfaces2023} constructed \emph{GCM hypersurfaces}, i.e. spacelike hypersurfaces foliated by GCM spheres, subject to further transversality and gauge constraints along the hypersurface:
\[
  (\div \eta)_{\ell=1} = 0, \qquad (\div \xib)_{\ell=1} = 0,
\]
see already Section \ref{sec:def-mass-charge} for the definitions of the Ricci coefficients $\eta$ and $\xib$.

The construction of GCM spheres and hypersurfaces proceeds by solving the determined system encoded in the GCM conditions, via a null frame transformation governed by transition functions $(\lambda, f, \fb)$, see \eqref{eq:Generalframetransf} for their exact form.
The transition functions $(\lambda, f, \fb)$ are chosen so that the transformed frame meets the prescribed constraints on $\trch$, $\trchb$, $\mu$, and the $\ell=1$ components such as $(\div \b)_{\ell=1}$.

In the proof of stability of slowly rotating Kerr, such a hypersurface is initiated from an intrinsic GCM sphere, whose location is quantitatively controlled via a dominance condition. The hypersurface then serves as the geometric backbone of the evolution: the gauge, determined by the GCM and transversality conditions, is fixed along this hypersurface, which acts as the future boundary of the domain under consideration. This “future-initialized” gauge is a key structural element of the proof, enabling control of the evolution backwards into the spacetime.

\subsubsection{Tracking the center of mass via GCM spheres}
A crucial feature of the vacuum case lies in the behavior of the
$\ell=1$ modes of $\div \b$. The full quantity $\div \b$ satisfies a
transport equation along $e_4$ involving the spin $+2$ curvature
component $\a$, whose decay is too weak to integrate
directly. Remarkably, $\a$ has no $\ell=1$ component, so the mode
$(\div \b)_{\ell=1}$ decouples from $\a$ and instead obeys a transport
equation with more favorable structure. Its backward evolution along
$e_3$ also exhibits sufficient decay to be integrated along $\Si_*$,
and still deduce the correct decay for it. This exceptional property
allows $(\div \b)_{\ell=1}$ to be controlled both along the GCM
hypersurface and in the exterior region $\Mext$.

Geometrically, the condition $(\div \b)_{\ell=1}=0$ identifies the center of mass of the final state and removes ambiguities from spatial translations. Analytically, controlling $(\div \b)_{\ell=1}$ establishes a hierarchy of estimates for curvature components and Ricci coefficients, and therefore the evolution of $(\div \b)_{\ell=1}$ plays a central role in the process.

This hierarchy underpins the nonlinear stability argument and ensures closure of the system under relatively weak decay assumptions.
According to \cite{klainermanCanonicalFoliationNull2024}, under suitable decay assumptions, $(\div \b)_{\ell=1}$—or an appropriate renormalization thereof—should converge at null infinity $\II^+$ to a well-defined limit, interpreted as the null center of mass of the spacetime. Any full nonlinear stability result must therefore guarantee the existence of this limit. Conversely, if the null center of mass cannot be controlled, one should not expect stability to hold.

\subsection{The charged case and the breakdown of the exceptional structure}

In the charged setting, new difficulties arise due to the inhomogeneous, non-autonomous behavior introduced by the Maxwell field. The most significant change concerns the $\ell=1$ modes. In the vacuum case, gauge-invariant gravitational perturbations of Schwarzschild or Kerr are supported only on $\ell \geq 2$ harmonics. In contrast, in the Einstein–Maxwell system, gauge-invariant coupled electromagnetic–gravitational perturbations around Reissner–Nordström or Kerr–Newman already involve nontrivial $\ell=1$ contributions.

This difference destroys the exceptional structure present in vacuum. There, the transport equation for $\nab_4 (\div \b)_{\ell=1}$ is structurally decoupled from the problematic spin-$+2$ curvature component $\a$, because $\a$ has no $\ell=1$ projection. As a result, $(\div \b)_{\ell=1}$ satisfies a more favorable transport equation with improved decay, allowing one to bypass the weak decay of the full $\nab_4 \b$ equation.

In the charged case, this decoupling fails: the Maxwell field produces source terms that couple directly to $(\div \b)_{\ell=1}$. Consequently, the $\nab_4$ equation inherits the same poor decay, and the key simplification of the vacuum case is lost. One may attempt to renormalize $(\div \b)_{\ell=1}$ to recover improved transport behavior, but such renormalizations inevitably improve only one of the two transport directions ($\nab_3$ or $\nab_4$), never both.\footnote{We have not found a renormalization with acceptable decay in both directions.}
This breakdown represents a fundamental obstruction to extending the GCM framework from vacuum to the charged setting.

From the analytic perspective, the renormalized center-of-mass quantity is indispensable, since it initiates the hierarchy of estimates for gauge-dependent fields. From the geometric perspective, this reflects the fact that while angular momentum is a true invariant, the center of mass is not: it can be shifted by a trivial coordinate transformation. Even for exact backgrounds, this freedom must be fixed. In the charged case, this necessitates a new gauge-fixing mechanism adapted to the electromagnetic structure of the equations.

\subsection{Mass-centered GCM spheres and hypersurfaces}

The failure of the exceptional structure in the charged case forces a departure from the standard GCM strategy of \cite{klainermanConstructionGCMSpheres2022, klainermanEffectiveResultsUniformization2022, shenConstructionGCMHypersurfaces2023}. In the vacuum setting, control of the center-of-mass quantity can be transported along a GCM hypersurface; in the charged setting this mechanism breaks down. Our proposal is to replace transport by a stronger condition: we impose the \emph{sphere-wise} vanishing of the center-of-mass quantity on each sphere foliating the hypersurface. We denote the center of mass by $\bm{C}_{\ell=1}$, see already \eqref{eq:definition-bm-C} for its definition. 

This also leads to the notion of a \emph{mass-centered GCM hypersurface}, defined as a spacelike hypersurface foliated by mass-centered GCM spheres. Because of the the additional center-of-mass condition, we are left with one less degree of freedom to impose on the hypersurface.

The idea of mass-centered GCM spheres thus bypasses the obstruction created by the lack of two-sided improved transport in the charged case, while retaining sufficient structure to implement the future-initialized gauge-fixing mechanism required for stability.
Moreover, the mass-centered GCM spheres are an alternative to the original GCM spheres of Klainerman-Szeftel, where \textit{each} sphere of the final slice $\Si_*$, and not only the first one, is required to identify the center-of-mass of the final black hole.

We give here a simplified version of our main theorems, see Theorems \ref{thm:GCM-spheres-J}, \ref{thm:intrinsic-GCM}, \ref{thm:GCMH}  for the precise versions. 

\begin{theorem}[Existence of mass-centered GCM spheres and hypersurfaces, First version]\label{thm:main-intro} Let $\RR$ be a fixed spacetime region verifying the assumptions given in Definition \ref{definition-spacetime-region-RR} and let $S$ be a fixed sphere of the foliation in $\RR$. Then 
  \begin{enumerate}
  \item for any $\underline{\Lambda} \in \mathbb{R}^3$ and any choice of approximated $\ell=1$ basis $J^{(p)}$, there exists a unique deformation of $S$ which is a \emph{mass-centered GCM sphere} ${\bf S}(\underline{\Lambda}, J^{(p)})$, i.e. in addition to 
    \eqref{eq:trch-trchb-mu}, it satisfies 
    \begin{equation}
      \bm{C}_{\ell=1,J^{(p)}} = 0.
    \end{equation}
  \item there exists a unique deformation of $S$, up to a rotation, which is an \emph{intrinsic mass-centered GCM sphere} ${\bf S}$, i.e. in addition to 
    \eqref{eq:trch-trchb-mu}, it satisfies 
    \begin{equation}
      (\trchbc)_{\ell = 1, J^{(p,\S)}} = 0, \qquad   \bm{C}_{\ell=1,J^{(p,\S)}} = 0,
    \end{equation}
    with respect to a canonical $\ell=1$ mode $J^{(p,\S)}$. 
  \item there exists a unique smooth spacelike \emph{mass-centered GCM hypersurface} $\Si$ passing through the unique mass-centered GCM sphere deformation of $S$, i.e. $\Sigma$ is foliated by mass-centered GCM spheres and satisfies in addition:
    \begin{align*}
      (\div\xib)_{\ell=1, J^{(p,\S)}}=0.
    \end{align*}
  \end{enumerate}
\end{theorem}

Theorems \ref{thm:GCM-spheres-J} and \ref{thm:intrinsic-GCM} correspond respectively to the modification of Theorem 1.1 in \cite{klainermanConstructionGCMSpheres2022} and Theorem 1.6 in \cite{klainermanEffectiveResultsUniformization2022}, and Theorem \ref{thm:GCMH} corresponds to Theorem 1.6 in \cite{shenConstructionGCMHypersurfaces2023}. Note that Theorem \ref{thm:main-intro} applies to the full subextremal range of charge and angular momentum. 

\begin{figure}[htbp]
  \centering
  \begin{tikzpicture}[scale=1.2, line cap=round, line join=round]

    \coordinate (A) at (-5,-2);
    \coordinate (B) at (3.5,-2.5);
    \coordinate (C) at (5,2.5);
    \coordinate (D) at (-3.5,3);

    \draw[blue!70, line width=1.2] (A)--(B)--(C)--(D)--cycle;

    \coordinate (A) at (-5,-2);
    \coordinate (B) at (3.5,-2.5);
    \coordinate (C) at (5,2.5);
    \coordinate (D) at (-3.5,3);

    \draw[blue!70, line width=1.2] (A)--(B)--(C)--(D)--cycle;

    \coordinate (Center) at (0,0);
    \coordinate (Ared) at ({cos(10)*(-5) - sin(10)*(-2)}, {sin(10)*(-5) + cos(10)*(-2)});
    \coordinate (Bred) at ({cos(10)*(3.5) - sin(10)*(-2.5)}, {sin(10)*(3.5) + cos(10)*(-2.5)});
    \coordinate (Cred) at ({cos(10)*(5) - sin(10)*(2.5)}, {sin(10)*(5) + cos(10)*(2.5)});
    \coordinate (Dred) at ({cos(10)*(-3.5) - sin(10)*(3)}, {sin(10)*(-3.5) + cos(10)*(3)});

    \draw[red!75, line width=1.2]
    (Ared) .. controls +(+0.2,0.2) and +(-0.2,0.2) .. (Bred)
    .. controls +(-0.2,0.2) and +(-0.2,-0.2) .. (Cred)
    .. controls +(-0.2,-0.2) and +(0.2,-0.2) .. (Dred)
    .. controls +(0.2,-0.2) and +(0.2,0.2) .. (Ared);

    \begin{scope}[rotate=12]
      \draw[red!75, line width=1.1] (0,0) ellipse (3.0 and 1.7);
      \draw[red!75, line width=1.1] (0,0) ellipse (2.1 and 1.1);
      \draw[red!75, line width=1.8] (0,0) ellipse (1.2 and 0.7);
    \end{scope}

    \begin{scope}[rotate=-2, shift={(0.3,0.1)}]
      \draw[blue!70, dashed, line width=1.1] (0,0) ellipse (3.0 and 1.7);
      \draw[blue!70, dashed, line width=1.1] (0,0) ellipse (2.1 and 1.1);
      \draw[blue!70, line width=1.8] (0,0) ellipse (1.2 and 0.7);
    \end{scope}


    \node[blue] at (2,-1.7) {$S$}; 
    \node[red]  at (-1.5,1.9) {$\S$}; 
    \node[red]  at (2.8,1.8) {$\bm{C}_{\ell=1,J^{(p)}} = 0$};
  \end{tikzpicture}
  \caption{The blue spheres represent the original foliation spheres $S$, while the red spheres ${\bf S}$ denote the deformed ones satisfying the vanishing center-of-mass condition.}
  \label{fig:placeholder}
\end{figure}
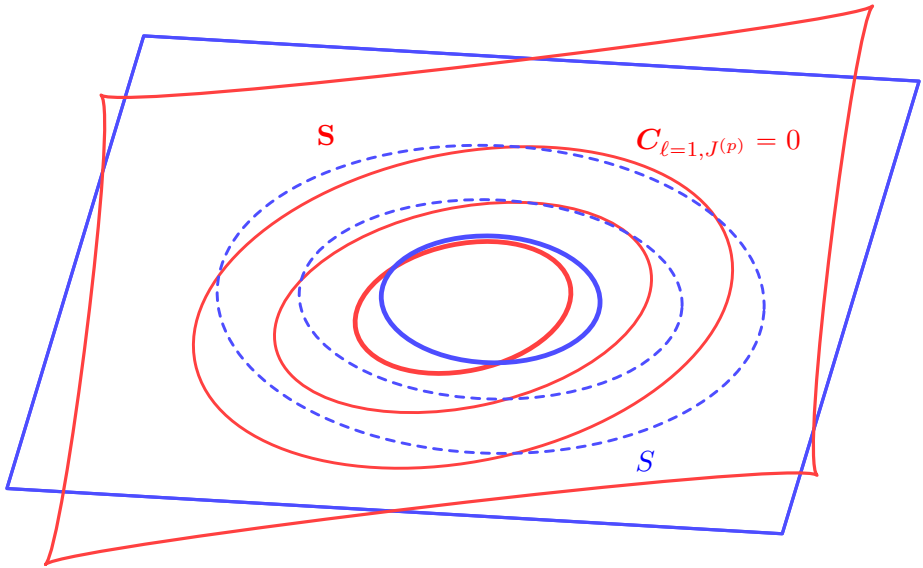

In our companion paper \cite{fangEinsteinMaxwellEquationsMassCentered2025}, we solve the Einstein-Maxwell equations on a mass-centered GCM hypersurface of the type constructed here.

\subsection{Main ideas behind the mass-centered GCM construction}

The proof of Theorem \ref{thm:main-intro} follows the same steps as the proof of Theorem 1.1 in \cite{klainermanConstructionGCMSpheres2022}, Theorem 1.6 in \cite{klainermanEffectiveResultsUniformization2022}, and Theorem 1.6 in \cite{shenConstructionGCMHypersurfaces2023} respectively, with important modification due to the change of the gauge conditions to the mass-centered case.

This modification imposes several requirements on the construction:
\begin{itemize}
\item The center-of-mass quantity $\bm{C}_{\ell=1}$ — defined as a suitable renormalization of $(\div \b)_{\ell=1}$ in the charged case—must reduce to $(\div \b)_{\ell=1}$ in the vacuum limit. Observe that both $\bm{C}_{\ell=1}$ and $(\div \b)_{\ell=1}$ are quantities well-defined on individual spheres.
\item The renormalized center-of-mass quantity must satisfy a favorable outgoing transport equation, so that control can still be propagated backward along outgoing null directions into the exterior region $\Mext$.
\item The hypersurface cannot coincide with that of \cite{shenConstructionGCMHypersurfaces2023}, since the available gauge freedom does not suffice to fix simultaneously the $\ell=1$ modes of $\div \eta$, $\div \xib$, and the center-of-mass quantity on each sphere.
\end{itemize}

We collect here a brief description of the proofs and modifications with respect to the corresponding Theorems in \cite{klainermanConstructionGCMSpheres2022, klainermanEffectiveResultsUniformization2022, shenConstructionGCMHypersurfaces2023}. 

\begin{enumerate}
\item The construction of \emph{mass-centered GCM spheres}, defined as
  GCM spheres satisfying the additional condition
  $\bm{C}_{\ell=1,J^{(p)}} = 0$ for a chosen basis $J^{(p)}$ of
  $\ell=1$ modes, proceeds in the same way as in
  \cite{klainermanConstructionGCMSpheres2022}. We employ the GCM
  system of adapted frame transformations and deformation in
  \cite{klainermanConstructionGCMSpheres2022}, see Section
  \ref{sec:GCM-system}, to deduce the existence of \emph{generic
    (charged) GCM spheres}, i.e. spheres satisfying
  \eqref{eq:trch-trchb-mu}, see Theorem \ref{thm:GCM-spheres}. The
  proof relies on an iterative scheme on the deformation map of a
  fixed sphere in the background as in
  \cite{klainermanConstructionGCMSpheres2022}, with the only
  modification being lower order terms in the charge of the
  coefficient functions of the GCM system.

  We specialize a generic charged GCM sphere to one which is also mass-centered by imposing the additional condition $\bm{C}_{\ell=1, J^{(p)}}=0$ in Theorem \ref{thm:GCM-spheres-J}. This is obtained by fixing the $(\div f)_{\ell=1, J^{(p)}}$ of the transition function through a nonlinear equation. Contrary to Theorem 1.6 in \cite{klainermanEffectiveResultsUniformization2022}, where both $(\div f)_{\ell=1}$ and $(\div \underline{f})_{\ell=1}$ are fixed for the construction of the intrinsic GCM sphere, here we only fix $(\div f)_{\ell=1}$ at this stage. 
\item 
  The construction of the \emph{intrinsic mass-centered GCM sphere} is obtained in Theorem \ref{thm:intrinsic-GCM} from the mass-centered GCM sphere of the previous point by fixing in addition $(\div\underline{f})_{\ell=1}$ to satisfy the additional constraint $(\trchb)_{\ell=1}=0$, analogously to the proof of Theorem 1.6 in \cite{klainermanEffectiveResultsUniformization2022}. Given a background sphere and a fixed basis $J^{(p)}$ of $\ell=1$ modes, the GCM conditions, together with $(\trchb)_{\ell=1}=0$ and $\bm{C}_{\ell=1}=0$, form a determined elliptic system for the frame transformation.

  The uniqueness of the intrinsic GCM sphere near a given background sphere is crucial, as it ensures that the construction is canonical. To make the procedure unambiguous, we employ the effective uniformization theorem of Klainerman–Szeftel \cite{klainermanEffectiveResultsUniformization2022} to select a canonical basis $J^{(p)}$ of $\ell=1$ modes. This eliminates residual supertranslation freedom and provides a geometrically meaningful cut at future null infinity $\II^+$, as emphasized in \cite{klainermanCanonicalFoliationNull2024} in the context of the Kerr stability result.

\item The \emph{mass-centered GCM hypersurface}, foliated by mass-centered GCM spheres, is constructed in Theorem \ref{thm:GCMH} by adapting the method of \cite{shenConstructionGCMHypersurfaces2023}. Compared with the framework in \cite{shenConstructionGCMHypersurfaces2023}, two conditions are now imposed: 
  \begin{align*}
    (\div \xib)_{\ell=1} = 0 \qquad \text{along the hypersurface,} \qquad \bm{C}_{\ell=1} = 0 \qquad \text{on each foliating sphere}.
  \end{align*}
  The condition $(\div \eta)_{\ell=1} = 0$ used in \cite{shenConstructionGCMHypersurfaces2023} is dropped, since the available gauge freedom does not suffice to fix all three quantities simultaneously. Nevertheless, the modified system remains sufficient to control the frame transformations along the hypersurface.

  As in \cite{shenConstructionGCMHypersurfaces2023}, the adapted frame is obtained from a general null frame via a transition function $(\lambda, f, \fb)$. The full set of GCM conditions, together with $(\div \xib)_{\ell=1}=0$ and $\bm{C}_{\ell=1}=0$, yields a determined elliptic–transport system along the hypersurface for solving this transition. The set up of the proof, involving the ODE system for the transition functions, follows closely the one in \cite{shenConstructionGCMHypersurfaces2023}, with the main difference being that it now only involves the $(\div \underline{f})_{\ell=1}$, as $(\div f)_{\ell=1}$ is fixed sphere-wise by the mass-centered condition along $\Sigma$. The mass-centered condition is imposed with respect to a family of $\ell=1$ modes $\widetilde{J}$ constructed by propagating the canonical one on the first sphere along $\Si$, as in \cite{shenConstructionGCMHypersurfaces2023}.
\end{enumerate}

\subsection{Organization of the paper}

The paper is organized as follows.
\begin{itemize}
\item Section 2 reviews the necessary preliminaries, including the background spacetime, $\ell=1$ modes, mass and charge parameters, linearized quantities, and general frame transformations. It also records the assumptions made for the construction and discusses deformation of surfaces and adapted $\ell=1$ modes.
  
\item Section 3 recalls the adapted frame transformation and deformation system developed by Klainerman–Szeftel, which serves as the foundation for both the previous and present GCM constructions.
  
\item Section 4 defines the mass-centered (charged) GCM spheres and proves their existence in Theorem \ref{thm:GCM-spheres-J}. We also include the construction of the intrinsic GCM sphere in Theorem \ref{thm:intrinsic-GCM}.
  
\item Section 5 constructs the mass-centered GCM hypersurface foliated by mass-centered (charged) GCM spheres. It states and proves Theorem \ref{thm:GCMH} regarding its existence.
\end{itemize}

\paragraph{Acknowledgments} A.J.F. acknowledges  support from the Deutsche
Forschungsgemeinschaft (DFG, German Research Foundation) through
Germany’s Excellence Strategy EXC 2044 390685587, Mathematics
M\"{u}nster: Dynamics–Geometry–Structure, from the Alexander von
Humboldt Foundation in the framework of the Alexander von Humboldt
Professorship endowed by the Federal Ministry of Education and
Research, and from NSF award DMS-2303241. E.G. acknowledges the support of NSF Grants DMS-2306143,
DMS-2336118 and of a grant of the Sloan Foundation. J.W. is supported by ERC-2023 AdG 101141855 BlaHSt.

\section{Preliminaries}

\subsection{Background spacetime and adapted coordinates}\label{sec:background}

We consider spacetime regions $\RR$ foliated by geodesic foliations $S(u,s)$ induced by an outgoing optical function $u$, with $s$ a properly normalized affine parameter along the null geodesic generators of $L=-\g^{\a\b}\partial_\b u \partial_\a$, where $\g$ is the spacetime metric.

We denote by $r=r(u,s)$ the area radius of $S(u,s)$, i.e such that the volume of $S(u,s)$ is given by $4\pi (r(u,s))^2$, and let $(e_3, e_4, e_1, e_2)$ be an adapted null frame where $e_4$ is proportional to $L$, $e_3$ is the unique null vectorfield orthogonal to $S(u,s)$ such that $\g(e_3, e_4)=-2$, and $e_1, e_2$ are tangent to the spheres $S=S(u,s)$.

Recall, see \cite{klainermanConstructionGCMSpheres2022}, that a
coordinate system $(u, s, y^1, y^2)$ is said to be \emph{adapted to an
outgoing geodesic foliation} $(e_3,e_4,e_1,e_2)$ if $e_4(y^1)= e_4(y^2)=0$. In
that case the spacetime metric can be written in the form
\begin{equation*}
  \g = - 2\vsi du ds + \vsi^2\Omb du^2 +g_{ab}\big( dy^a- \vsi \undB^a du\big) \big( dy^b-\vsi \undB^b du\big),
\end{equation*}
where
\begin{equation*}
  \Omb=e_3(s), \qquad \undB^a =\frac{1}{2} e_3(y^a), \qquad g_{ab}=\g(\pr_{y^a}, \pr_{y^b}).
\end{equation*}
Relative to these coordinates
\begin{align*}
  e_4=\pr_s, \qquad \pr_u = \vsi\left(\frac{1}{2}e_3-\frac{1}{2}\Omb e_4-\undB^a\pr_{y^a}\right), \qquad \pr_{y^a}= \sum_{c=1,2} Y_{(a)}^c e_c, \qquad a=1,2,
\end{align*}
with coefficients $ Y_{(a)}^b $ verifying $g_{ab}=\sum_{c=1, 2} Y_{(a)}^c Y_{(b)}^c$.

As in \cite{klainermanConstructionGCMSpheres2022}, we assume that $\RR $ is covered by two coordinate systems, i.e. $\RR=\RR_N\cup \RR_S$,
such that
\begin{enumerate}
\item The North coordinate chart $\RR_N$ is given by the coordinates
  $(u, s, y_{N}^1, y_{N}^2)$ with $(y^1_{N})^2+(y^2_{N})^2<2$. 

\item The South coordinate chart $\RR_S$ is given by the coordinates
  $(u, s, y_{S}^1, y_{S}^2)$ with $(y^1_{S})^2+(y^2_{S})^2<2$. 

\item The two coordinate charts intersect in the open equatorial region
  $\RR_{Eq}\vcentcolon=\RR_N\cap \RR_S$ in which both coordinate systems are defined.
  
\item In $\RR_{Eq} $ the transition functions between the two coordinate systems are given by the smooth functions $ \varphi_{SN}$ and $\varphi_{NS}= \varphi_{SN}^{-1} $. 
\end{enumerate}
The metric coefficients for the two coordinate systems are given by
\begin{align*}
  \g^N &= - 2\vsi du ds + \vsi^2\Omb du^2 +g^{N}_{ab}\big( dy_N^a- \vsi \undB_{N}^a du\big) \big( dy_N^b-\vsi \undB_N^b du\big),\\
  \g^S &= - 2\vsi du ds + \vsi^2\Omb du^2 +g^{S}_{ab}\big( dy_S^a- \vsi \undB_{S}^a du\big) \big( dy_S^b-\vsi \undB_S^b du\big),
\end{align*}
where
\begin{align*}
  \Omb=e_3(s), \qquad \undB_N^a =\frac{1}{2} e_3(y_N^a), \qquad \undB_S^a =\frac{1}{2} e_3(y_S^a).
\end{align*}
On a 2-sphere $S(u,s)$ we consider the following norms, 
\begin{equation*}
  \begin{split}
    \|f\|_{\infty,k} :&= \sum_{i=0}^k \|\dk^i f\|_{L^\infty(S) }, \qquad 
                        \|f\|_{\mathfrak{h}_k}\vcentcolon=\sum_{i=0}^k \|\dkb^i f\|_{L^2(S)},
  \end{split}
\end{equation*}
where $\dk^i$ stands for any combination of length $i$ of operators of the form 
$e_3, r e_4, r\nab $ and $\dkb^i$ stands for any combination of length $i$ of $r\nab$.

\subsection{Basis of \texorpdfstring{$\ell=1$}{} modes}\label{sec:l=1modes}

Recall that on the standard round sphere $\mathbb{S}^2$ with sperical coordinates $(\theta, \phi)$, the basis of the $\ell=1$ spherical harmonics is given by
\begin{align}\label{eq:sph-harm}
  J^{(\mathbb{S}^2, 0)}=\cos\th, \qquad J^{(\mathbb{S}^2, +)}=\sin\th\cos\phi, \qquad J^{(\mathbb{S}^2, -)}=\sin\th\sin\phi.
\end{align}
Given a small parameter $\mathring{\epsilon}>0$, we generalize the above definition for any sphere $S$ as follows.\footnote{The property of the scalar functions $\Jp$ above is motivated by the fact that the $\ell=1$ spherical harmonics on the standard sphere $\SSS^2$, given by $J^{(0, \SSS^2)}=\cos\th, \, J^{(+, \SSS^2)}=\sin\th\cos\vphi, \, J^{(-, \SSS^2)}=\sin\th\sin\vphi$, 
  satisfy \eqref{eq:Jpsphericalharmonics} with $\epg=0$. Note also that on $\SSS^2$,
  \begin{align*}
    \int_{\mathbb{S}^2}(\cos\th)^2=\int_{\mathbb{S}^2}(\sin\th\cos\vphi)^2=\int_{\mathbb{S}^2}(\sin\th\sin\vphi)^2=\frac{4\pi}{3}, \qquad |\SSS^2|=4\pi.
  \end{align*},} 
\begin{definition}\label{def:l=1}
  We say that a triplet of smooth real scalar functions $J^{(p)}=\{J^{(0)}, J^{(+)}, J^{(-)}\} $ on $S\subset \RR$ is an \emph{$\epg$-approximated basis of $\ell=1$ modes} on $S$ if the following are verified:
  \begin{equation}\label{eq:Jpsphericalharmonics}
    \begin{split}
      \left(r^2 \bigtriangleup+2\right)J^{(p)}=O(\epg), \qquad p=0,+,-,\\
      \frac{1}{|S|}\int_{S}J^{(p)}J^{(q)}=\frac 13 \delta_{pq} +O(\epg), \qquad p,q=0,+,-,\\
      \frac{1}{|S|}\int_{S}J^{(p)}=O(\epg), \qquad p=0,+,- , 
    \end{split}
  \end{equation}
  where $\epg>0$ is a sufficiently small constant.
\end{definition}

Observe that for a basis of $\ell=1$ mode as in Definition \ref{def:l=1}, we have
\begin{align*}
  \dds_2\dds_1 J^{(p)}=r^{-2}O(\mathring{\epsilon}).
\end{align*}
Assuming the existence of such a basis $J^{(p)}$, we define for a scalar function $\lambda$ on $S$:
\begin{align}
  \lambda_{\ell=1, J}\vcentcolon=\biggl\{ \int_{S}J^{(p)} \lambda, \qquad p\in \{-, 0, +\}\biggl\},
\end{align}
and for a 1-form $f$ on $S$:
\begin{align}\label{eq:defell=1foroneformofepgspheres}
  f_{\ell=1, J}\vcentcolon=\biggl\{ \int_{S}J^{(p)} \ddd_1^{S}f, \qquad p\in \{-, 0, +\}\biggl\},
\end{align}
where $\ddd_1^{S}f=(\div^S f, \curl^S f)$.

We say that a function $\lambda$ is supported in $\ell\leq 1$, denoted by $\lambda_{\ell\geq 2, J}=0$, if there exist constants $A_0$, $B_{-}$, $B_{0}$, $B_{+}$ such that 
\begin{align*}
  \lambda= A_0 + B_{-}J^{(-)}+B_{0}J^{(0)}+B_{+}J^{(+)}.
\end{align*}

We also use the following notation for a function $\lambda$ to denote its projection to the $\ell\geq1$ modes:
\begin{align}\label{eq:ellgeq1}
  \lambda_{\ell\geq 1}\vcentcolon=\lambda - \frac{1}{|S|}\int_S \lambda,
\end{align}
which is independent of the choice of $J^{(p)}$.

Recall from \cite{klainermanEffectiveResultsUniformization2022} that there is a way of defining canonical
\texorpdfstring{$\ell=1$}{} modes on spheres. We first recall
the effective uniformization theorem of \cite{klainermanEffectiveResultsUniformization2022}.

\begin{theorem}[Effective Uniformization; Corollary 3.8 in \cite{klainermanEffectiveResultsUniformization2022}]
  \label{thm:eff-unif}
  Let $(S, g^S)$ be a fixed sphere. Then, up to isometries of $\SSS^2$, there exists a unique diffeomorphism $\Phi:\SSS^2 \to S$ and a unique centered conformal factor $u \in H^1(\SSS^2)$ with
  \[
    \int_{\SSS^2} x e^{2u} = 0,
  \]
  such that
  \[
    \Phi^\#(g^S) = (r^S)^2 e^{2u} \ga_0.
  \]
  Moreover, under the almost round condition
  \begin{equation}
    \label{eq:almost-round}
    \left\|K^S - \frac{1}{(r^S)^2} \right\|_{L^\infty(S)} \le \frac{\ep}{(r^S)^2},
  \end{equation}
  where $K^S = K(g^S)$ is the Gauss curvature, we have 
  \begin{equation*}
    \| u \circ \Phi^{-1} \|_{L^\infty(S)} \les \ep.
  \end{equation*}
  for sufficiently small $\ep > 0$.

\end{theorem}

\begin{definition}[Basis of canonical $\ell=1$ modes on $S$]
  \label{def:can-ell1-modes}
  Let $(S, g^S)$ be an almost round sphere, i.e., satisfying \eqref{eq:almost-round}. Let $(\Phi, u)$ be the unique, up to isometries of $\SSS^2$, uniformization pair given by Theorem \ref{thm:eff-unif}, so that
  \[
    \Phi : \SSS^2 \longrightarrow S, \qquad \Phi^\#(g^S) = (r^S)^2 e^{2u} \ga_0, \qquad u \in H^1(\SSS^2), \quad \int_{\SSS^2} x e^{2u} = 0.
  \]
  We define the \emph{canonical $\ell=1$ modes on $S$} by
  \[
    J^S \vcentcolon= J^{\SSS^2} \circ \Phi^{-1},
  \]
  where $J^{\SSS^2}$ denotes the $\ell=1$ spherical harmonics as in \eqref{eq:sph-harm}.
\end{definition}

\begin{remark}
  The canonical basis $J^S$ is unique up to a rotation on $\SSS^2$.
\end{remark}

\subsection{Definition of mass and charge parameters}\label{sec:def-mass-charge}

We recall the following definition of the corresponding connection coefficients and null Weyl curvature components relative to the null frame $(e_3, e_4, e_1, e_2)$:
\begin{align*}
  \chib_{ab}&=\g(\D_{e_a}e_3, e_b), \qquad \chi_{ab}=\g(\D_{e_a}e_4, e_b), \\
  \xib_a&=\frac 1 2 \g(\D_3 e_3, e_a), \qquad \xi_a=\frac 1 2 \g(\D_4 e_4, e_a), \\
  \omb&=\frac 1 4 \g(\D_3 e_3, e_4), \qquad \om=\frac 1 4 \g(\D_4 e_4, e_3),\\
  \etab_a&=\frac 1 2 \g( \D_4 e_3, e_a), \qquad \eta_a=\frac 1 2 \g( \D_3 e_4, e_a),\\
  \ze_a&=\frac 1 2\g(\D_{e_a}e_4, e_3),
\end{align*}
and
\begin{align*}
  \alpha_{ab}&=\W(e_a, e_4, e_b, e_4), \qquad \aa_{ab}=\W(e_a, e_3, e_b, e_3),\\
  \b_a&=\frac 1 2 \W(e_a, e_4, e_3, e_4), \qquad \bb_a= \frac 1 2\W(e_a, e_3, e_3, e_4), \\
  \rho&=\frac 1 4 \W(e_3, e_4, e_3, e_4), \qquad \dual\rho=\frac 1 4 \dual\W (e_3, e_4, e_3, e_4),
\end{align*}
where $\W$ is the Weyl curvature of the spacetime metric $\g$ and $\dual \W$ is its Hodge dual.

We denote
\begin{align*}
  \trch\vcentcolon=\de^{ab}\chi_{ab}, \qquad \trchb\vcentcolon=\de^{ab}\chib_{ab}.
\end{align*}

Assume that the spacetime region $\RR$ comes endowed with a closed and co-closed two form $\F$, the electromagnetic tensor, i.e.
\begin{align}
  \div \F=0, \qquad d\F=0.
\end{align}
We define the following null electromagnetic components relative to the null frame $(e_3, e_4, e_1, e_2)$:
\begin{align*}
  \bF_a&= \F(e_a, e_4), \qquad \bbF_a= \F(e_a, e_3), \\
  \rhoF&=\frac 1 2 \F(e_3, e_4), \qquad \dual\rhoF=\frac12\dual\F (e_3, e_4).
\end{align*}

\begin{remark}
  Observe that in an outgoing geodesic foliation of Reissner-Nordstr\"om spacetime, we have
  \begin{align}\label{eq:condition-RN}
    \trch =\frac{2}{r}, \qquad
    \trchb=-\frac{2}{r}\left( 1-\frac{2M}{r}+\frac{Q^2}{r^2}\right), \qquad
    \mu= \frac{2M}{r^3}-\frac{2Q^2}{r^4}.
  \end{align}
  In the charged Reissner-Nordstr\"om spacetime, the standard definition of electric and magnetic charge and that of Hawking mass of a sphere $S$ or radius $r$ are respectively
  \begin{align*}
    Q^S \vcentcolon= \frac{1}{4\pi}\int_{S} \rhoF, \qquad
    e^S \vcentcolon= \frac{1}{4\pi}\int_{S} \dual \rhoF, \qquad
    \frac{2m^S}{r}\vcentcolon=1+\frac{1}{16}\int_S \trch \trchb,
  \end{align*}
  and result\footnote{The relation $e^S=0$ can be enforced by making use of the $U(1)$ symmetry of the Maxwell equations.} in the relation
  \begin{align}
    m^S=M-\frac{(Q^S)^2}{2r}, \qquad e^S=0.
  \end{align}
\end{remark}

The above remark motivates the following definitions for a sphere $S\subset \RR$. 
\begin{definition}\label{def:parameters} On any sphere $S\subset \RR$ with area radius $r$, we define the following parameters.
  \begin{enumerate}
  \item The \emph{auxiliary mass $m^S$} is chosen to be the Hawking mass of $S$, i.e.,
    \begin{align}\label{eq:definition-auxiliary-m}
      \frac{2m^S}{r}\vcentcolon=1+\frac{1}{16\pi} \int_{S} \trch\trchb. 
    \end{align}
  \item The \emph{electric charge $Q^S$} is chosen as the average of $\rhoF$ at $S$, i.e.
    \begin{equation}
      \label{eq:M3:Q:def}
      Q^S \vcentcolon= \frac{1}{4\pi}\int_{S} \rhoF= \frac{1}{4\pi}\int_{S}\dual\F.
    \end{equation}
  \item The \emph{magnetic charge $e^S$} is chosen as the average of $\dual\rhoF$ at $S$, i.e.
    \begin{equation}
      e^S \vcentcolon= \frac{1}{4\pi}\int_{S} \dual \rhoF= \frac{1}{4\pi}\int_{S}\F.
    \end{equation}
  \item The \emph{mass parameter $M^S$} of $S$ is chosen as 
    \begin{align}
      M^S \vcentcolon= m^S + \frac{(Q^S)^2+(e^S)^2}{2r}.
    \end{align}
  \item The \emph{center of mass function} on $S$ is defined as 
    \begin{equation}
      \label{eq:definition-bm-C}
      \bm{C}^S\vcentcolon=\div\b -\frac{Q^S}{r^2}\div \bF+\frac{2Q^S}{r^3}\rhoFc,
    \end{equation}
    where $\rhoFc \vcentcolon= \displaystyle\rhoF -\frac{Q^S}{r^2}$.
  \end{enumerate}
\end{definition}

\begin{remark}
  An analogous definition for the angular momentum function of $S$ can be given, see \cite{fangEinsteinMaxwellEquationsMassCentered2025}.
\end{remark}

Recall that in the above $r$ always denotes the area radius $r\vcentcolon=\sqrt{\frac{|S|}{4\pi}}$. We also use the notation:
\begin{align}
  \Up\vcentcolon= 1-\frac{2M}{r}+\frac{Q^2+e^2}{r^2}=1-\frac{2m}{r}.
\end{align}

\begin{remark}
  In the case of vacuum spacetime, $Q=0$ and the center of mass function $\bm{C}$ reduces to $\div\beta$, as used in \cite{klainermanEffectiveResultsUniformization2022}.
\end{remark}
\begin{remark}
  By the definition of electric charge \eqref{eq:M3:Q:def}, notice that $\rhoFc$ is defined to be the difference between $\rhoF$ and its average, and therefore by \eqref{eq:ellgeq1} we have $\rhoFc= \rhoF_{\ell\geq 1}$.
\end{remark}

\subsection{Linearized quantities}

The mass and charge parameters given in Definition \ref{def:parameters} are used to define the linearized quantities with respect to their expected value in Reissner-Nordstr\"om. 

\begin{definition} 
  We define the following linearized quantities:
  \begin{gather*}
    \trchc \vcentcolon= \ds\trch-\frac{2}{r},\qquad
    \trchbc \vcentcolon= \ds\trchb+\frac{2\Up}{r}, \qquad \ombc \vcentcolon= \ds\omb-\frac{M}{r^2}+\frac{Q^2}{r^3}, \\
    \Kc \vcentcolon= K -\frac{1}{r^2}, \qquad 
    \rhoc \vcentcolon= \ds \rho +\frac{2M}{r^3} - \frac{2Q^2}{r^4}, \qquad \widecheck{\mu} \vcentcolon= \ds \mu -\frac{2M}{r^3} + \frac{2Q^2}{r^4}, \\
    \rhoFc \vcentcolon= \displaystyle\rhoF -\frac{Q}{r^2}, \qquad \widecheck{\Omb} \vcentcolon=\Omb+\Up, \qquad \widecheck{\varsigma} \vcentcolon= \varsigma-1,
  \end{gather*}
  where $M$ and $Q$ are the mass and electric charge defined in Definition \ref{def:parameters} and $\Up = 1-\frac{2M}{r} + \frac{Q^2}{r^2}$.
\end{definition}

We divide the linearized quantities into two sets, denoted $\Gag, \Gab$, depending on their expected decay in $r$: 
\begin{equation}\label{eq:Definition-Gag-Gab}
  \begin{split}
    \Gag\vcentcolon={}& \Bigg\{\trchc,\,\, \chih, \,\, \ze, \,\, \trchbc,\,\, r\muc ,\,\, r\rhoc, \,\, r\dual\rho, \,\, r\b, \,\, r\a, \,\,\bF, \,\,\dual\rhoF, \,\, \rhoFc,\,\,\\
            &  r\Kc, \,\, r^{-1} \big(e_4(r)-1\big),\,\, r^{-1}e_4(m), \,\, r^{-2}e_4(Q)\Bigg\},\\
    \Ga_b \vcentcolon={}& \Bigg\{\eta, \,\,\chibh, \,\, \ombc, \,\, \xib,\,\, r\bb, \,\, \aa, \,\, \bbF, \,\, r^{-1}\Ombc, \,\,r^{-1}\widecheck{\varsigma}, \,\, r^{-1}(e_3(r)+\Up\big), \,\, r^{-1}e_3(m), \,\, r^{-2} e_3(Q) \Bigg\}. 
  \end{split}
\end{equation}


\subsection{Main assumptions for \texorpdfstring{$\RR$}{}}
\label{subsubsect:regionRR2}


In the following definition, we specify the background spacetime region $\RR$.

\begin{definition}\label{definition-spacetime-region-RR}
  Let $m_0, Q_0>0$ be constants. Let $\epg>0$ be a sufficiently small constant and let $(\ug, \sg, \rg)$ three real numbers with $\rg$ sufficiently large so that 
  \begin{align}
    \epg\ll m_0, Q_0\ll \rg.
  \end{align}
  We define $\RR$ to be the region
  \begin{align}
    \RR\vcentcolon= \{ |u-\ug| \leq \epg, \qquad |s-\sg| \leq \epg \}
  \end{align}
  such that the following assumptions {\bf A1-A4} below are verified with constant $\epg$ on the background foliation of $\RR$.

  Given an integer $s_{max}\geq 3$, we assume\footnote{In view of \eqref{eq:assumtioninRRforGagandGabofbackgroundfoliation}, we will often replace $\Ga_g$ by $r^{-1} \Ga_b$.} the following. 
  \begin{enumerate}
  \item[\bf A1.]
    For $k\le s_{max}$
    \begin{equation}\label{eq:assumtioninRRforGagandGabofbackgroundfoliation}
      \begin{split}
        \| \Ga_g\|_{\infty,k}&\leq \epg r^{-2},\\
        \| \Ga_b\|_{\infty,k}&\leq \epg r^{-1}.
      \end{split}
    \end{equation}

  \item[\bf A2.] The Hawking mass $m=m(u,s)$, the electric charge $Q=Q(u,s)$ and the magnetic charge $e=e(u,s)$ of $S(u, s)$ verify 
    \begin{equation}\label{eq:assumtionsonthegivenusfoliationforGCMprocedure:Hawkingmass} 
      \sup_{\RR}\left|\frac{m}{m_0}-1\right| \leq \epg,\qquad \sup_{\RR}\left|\frac{Q}{Q_0}-1\right| \leq \epg, \qquad \sup_{\RR}|e|\leq \epg.
    \end{equation}

  \item[\bf A3.] 
    In the region of their respective validity\footnote{That is the quantities on the left verify the same estimates as those for $\Ga_b$, respectively $\Ga_g$.} we have
    \begin{equation}
      \undB_N^a,\,\, \undB_S^a \in r^{-1}\Ga_b, \qquad Z_N^a,\,\, Z_S^a \in \Ga_b,\qquad r^{-2} \widecheck{g}^{N}_{ab}, \,\, r^{-2} \widecheck{g}^{S}_{ab} \in r\Ga_g,
    \end{equation}
    where
    \begin{align*}
      \widecheck{g} ^{N}\!_{ab} &= g^N_{ab} - \frac{4r^2}{1+(y^1_{N})^2+(y^2_{N})^2) }\de_{ab},\\
      \widecheck{g}^{S}\!_{ab} &= g^S_{ab} - \frac{4r^2}{(1+(y^1_{S})^2+(y^2_{S})^2) } \de_{ab}.
    \end{align*}

  \item[\bf A4.] We assume the existence of a smooth family of scalar functions
    $\Jp:\RR\to\RRR$, for $p=0,+,-$, which is an $\epg$-approximated basis of $\ell=1$ modes as in Definition \ref{def:l=1} for all surfaces $S$ of the background foliation.
  \end{enumerate}
\end{definition}

\begin{remark}
  We note that the assumptions {\bf A1}, {\bf A2}, {\bf A3}, {\bf A4},
  are expected to be valid in the asymptotically flat regions of
  perturbations of a Kerr-Newman black hole.
\end{remark}

What follows is mostly independent of the
Einstein(-Maxwell) equations. We summarize here what
consequences of the Einstein(-Maxwell) equation we use in the following.

\paragraph{Assumptions on $e_4$ derivatives of some GCM quantities.} We need to assume that
\begin{align}\label{eq:assumption-Einstein}
  \begin{split}
    \|  e_4 (\trchc) \|_{\mathfrak{h}_k} = r^{-3}O(\epg), \\
    \|  e_4(\trchbc)\|_{\mathfrak{h}_k}= r^{-3}O(\epg), \\
    \| e_4 (\muc)\|_{\mathfrak{h}_k} = r^{-4}O(\epg).
  \end{split}
\end{align}

\begin{remark}
  The above are obviously satisfied in the case of Einstein vacuum or Einstein-Maxwell equations as a consequence of assumption {\bf A1}. They would also be satisfied in most reasonable matter models. 
\end{remark}

\subsection{General frame transformations}


\begin{lemma}[Lemma 3.1 in \cite{klainermanConstructionGCMSpheres2022}]
  \label{Lemma:Generalframetransf}
  Given a null frame $(e_3, e_4, e_1, e_2)$, a general null transformation from the null frame $(e_3, e_4, e_1, e_2)$ to another null frame $(e_3', e_4', e_1', e_2')$ can be written in the form,
  \begin{equation}\label{eq:Generalframetransf}
    \begin{split}
      e_4'&=\la\left(e_4 + f^b e_b +\frac 1 4 |f|^2 e_3\right),\\
      e_a'&= \left(\de_{ab} +\frac{1}{2}\fb_af_b\right) e_b +\frac 1 2 \fb_a e_4 +\left(\frac 1 2 f_a +\frac{1}{8}|f|^2\fb_a\right) e_3,\qquad a=1,2,\\
      e_3'&=\la^{-1}\left( \left(1+\frac{1}{2}f\c\fb +\frac{1}{16} |f|^2 |\fb|^2\right) e_3 + \left(\fb^b+\frac 1 4 |\fb|^2f^b\right) e_b + \frac 1 4 |\fb|^2 e_4 \right),
    \end{split}
  \end{equation}
  where $\la$ is a scalar, $f$ and $\fb$ are horizontal 1-forms, i.e.
  \begin{equation*}
  f(e_3) = f(e_4) = \fb(e_3) = \fb(e_4)=0.
  \end{equation*}
  The dot product and magnitude $|\c |$ are taken with respect to the
  standard euclidian norm of $\mathbb{R}^2$. We call $(f, \fb, \la)$
  the transition coefficients of the change of frame. We denote
  $F\vcentcolon=(f, \fb, \ovla)$ where $\ovla=\la-1$.
\end{lemma}

The Ricci coefficients, curvature and electromagnetic components get modified by a general frame transformation as above. We recall here the frame transformations which will be useful in the following. For a full list of transformation identities see for example Proposition 2.14 in \cite{shenConstructionGCMHypersurfaces2023}.

\begin{lemma}[Proposition 3.3 in \cite{klainermanConstructionGCMSpheres2022}]\label{lem:EM-transform}
  Under a general transformation of type \eqref{eq:Generalframetransf}, the Ricci coefficients transform as follows:
  \begin{itemize}
  \item The transformation formulas for $\chi $ are given by 
    \begin{equation*}
      \begin{split}
        \la^{-1}\trch' ={}& \trch + \div'f + f\c\eta + f\c\ze+\err(\trch,\trch')\\
        \err(\trch,\trch') ={}& \fb\c\xi+\frac{1}{4}\fb\c\left(f\trch -\dual f\atrch\right) +\om (f\c\fb) -\omb |f|^2 \\
                       & -\frac{1}{4}|f|^2\trchb - \frac{1}{4} ( f\c\fb) \la^{-1}\trch' +\frac{1}{4} (\fb\wedge f) \la^{-1}\atrch'+\lot,
      \end{split}
    \end{equation*}

  \item The transformation formula for $\xi$ is given by 
    \begin{equation*}
      \begin{split}
        \la^{-2}\xi' ={}& \xi +\frac{1}{2}\la^{-1}\nab_4'f+\frac{1}{4}(\trch f -\atrch\dual f)+\om f +\err(\xi,\xi'),\\
        \err(\xi,\xi') ={}& \frac{1}{2}f\c\chih+\frac{1}{4}|f|^2\eta+\frac{1}{2}(f\c \ze)\,f -\frac{1}{4}|f|^2\etab \\
                     &+ \la^{-2}\left( \frac{1}{2}(f\c\xi')\,\fb+ \frac{1}{2}(f\c\fb)\,\xi' \right) +\lot
      \end{split}
    \end{equation*}

  \item The transformation formula for $\xib$ is given by 
    \begin{equation*}
      \begin{split}
        \la^2\xib' &= \xib + \frac{1}{2}\la\nab_3'\fb + \omb\,\fb + \frac{1}{4}\trchb\,\fb - \frac{1}{4}\atrchb\dual\fb +\err(\xib, \xib'),\\
        \err(\xib, \xib') &= \frac{1}{2}\fb\c\chibh - \frac{1}{2}(\fb\c\ze)\fb + \frac{1}{4} |\fb|^2\etab -\frac{1}{4} |\fb|^2\eta'+\lot
      \end{split}
    \end{equation*}

  \item The transformation formula for $\eta$ is given by 
    \begin{equation*}
      \begin{split}
        \eta' &= \eta +\frac{1}{2}\la \nab_3'f +\frac{1}{4}\fb\trch -\frac{1}{4}\dual\fb\atrch -\omb\, f +\err(\eta, \eta'),\\
        \err(\eta, \eta') &= \frac{1}{2}(f\c\fb)\eta +\frac{1}{2}\fb\c\chih
                            +\frac{1}{2}f(\fb\c\ze) - (\fb\c f)\eta'+ \frac{1}{2}\fb (f\c\eta') +\lot
      \end{split}
    \end{equation*}
  \item The transformation formula for $\etab$ is given by 
    \begin{equation*}
      \begin{split}
        \etab' &= \etab +\frac{1}{2}\la^{-1}\nab_4'\fb +\frac{1}{4}\trchb f - \frac{1}{4}\atrchb\dual f -\om\fb +\err(\etab, \etab'),\\
        \err(\etab, \etab') &= \frac{1}{2}f\c\chibh + \frac{1}{2}(f\c\etab)\fb-\frac{1}{4} (f\c\ze)\fb -\frac{1}{4} |\fb|^2\la^{-2}\xi'+\lot
      \end{split}
    \end{equation*}

  \item The transformation formula for $\om$ is given by
    \begin{equation*}
      \begin{split}
        \la^{-1}\om' &= \om -\frac{1}{2}\la^{-1}e_4'(\log\la)+\frac{1}{2}f\c(\ze-\etab) +\err(\om, \om'),\\
        \err(\om, \om') &= -\frac{1}{4}|f|^2\omb - \frac{1}{8}\trchb |f|^2+\frac{1}{2}\la^{-2}\fb\c\xi' +\lot
      \end{split}
    \end{equation*}

  \item The transformation formula for $\omb$ is given by
    \begin{equation*}
      \begin{split}
        \la\omb' ={}& \omb+\frac{1}{2}\la e_3'(\log\la) -\frac{1}{2}\fb\c\ze -\frac{1}{2}\fb\c\eta +\err(\omb,\omb'),\\
        \err(\omb,\omb') ={}& f\c\fb\,\omb-\frac{1}{4} |\fb|^2\om +\frac{1}{2}f\c\xib + \frac{1}{8}(f\c\fb)\trchb + \frac{1}{8}(\fb\wedge f)\atrchb \\
                 & -\frac{1}{8}|\fb|^2\trch -\frac{1}{4}\la \fb\c\nab_3'f +\frac{1}{2} (\fb\c f)(\fb\c\eta')- \frac{1}{4}|\fb|^2 (f\c\eta')+\lot
      \end{split}
    \end{equation*}
  \end{itemize}
  where, for the transformation formulas of the Ricci coefficients above, $\lot$ denote expressions of the type
  \begin{eqnarray*}
    \lot&=O((f,\fb)^3)\Ga +O((f,\fb)^2) \Gac
  \end{eqnarray*}
  containing no derivatives of $f$, $\fb$, $\Ga$ and $\Gac$.

  Under a general null frame transformation of the form \eqref{eq:Generalframetransf}, the electromagnetic components transform as follows:
  \begin{itemize}
  \item The transformation formula for $\bF$, $\bbF$ are given by
    \begin{align*}
      \lambda^{-1} \bF' &= \bF - \dual f \, \dual\rhoF + f \, \rhoF + \lot, \\
      \lambda \, \bbF' &= \bbF - \dual \fb \, \dual\rhoF - \fb \, \rhoF + \lot.
    \end{align*}

  \item The transformation formula for $\rhoF$ and $\dual\rhoF$ are given by
    \begin{align*}
      \rhoF' &= \rhoF + \err(\rhoF, \rhoF'), \\
      \err(\rhoF, \rhoF') &= -\frac{1}{2} f \c \bbF + \frac{1}{2} \fb \c \bF + \frac{1}{2} \rhoF (f \c \fb) + \frac{1}{2} \dual\rhoF (f \wedge \fb) + \lot, \\
      \dual\rhoF' &= \dual\rhoF + \err(\dual\rhoF, \dual\rhoF'), \\
      \err(\dual\rhoF, \dual\rhoF') &= -\frac{1}{2} f \c \dual \bbF - \frac{1}{2} \fb \c \dual \bF + \frac{1}{2} \dual\rhoF (f \c \fb) - \frac{1}{2} \rhoF (f \wedge \fb) + \lot.
    \end{align*}
  \end{itemize}

  Here, the error term $\lot$ denotes cubic and quadratic terms in $f$, $\fb$:
  \[
    \lot = O((f,\fb)^3)(\rhoF, \dual\rhoF) + O((f,\fb)^2)(\bF, \bbF),
  \]
  and contains no derivatives of $f$, $\fb$, $\bF$, $\bbF$, $\rhoF$, or $\dual\rhoF$.
\end{lemma}

\begin{proof}
  The transformation formulas for the Ricci coefficients are the same as in the case of vacuum.
  We compute the formulas for the electromagnetic components:
  \begin{align*}
    \bF'_a &= \F(e'_a, e'_4) \\
           &= \lambda \F\left( \left( \de_a^b + \frac{1}{2} \fb_a f^b \right) e_b + \frac{1}{2} \fb_a e_4 + \frac{1}{2} f_a e_3, \; e_4 + f^b e_b + \frac{1}{4} |f|^2 e_3 \right) + \lot \\
           &= \lambda \left( \F_{a4} + f^b \F_{ab} + \frac{1}{2} f_a \F(e_3, e_4) \right) + \lot \\
           &= \lambda \left( \bF_a - \dual f_a \dual\rhoF + f_a \rhoF \right) + \lot.
  \end{align*}
  The formula for $\bbF$ follows by symmetry, using that $\rhoF$ changes sign and $\dual\rhoF$ remains unchanged under $e_3 \leftrightarrow e_4$. Next, compute:
  \begin{align*}
    \rhoF' &= \frac{1}{2} \F(e_3', e_4') \\
           &= \frac{1}{2} \F\left( \left(1 + \frac{1}{2} f \c \fb \right) e_3 + \fb^b e_b + \frac{1}{4} |\fb|^2 e_4, \; e_4 + f^b e_b + \frac{1}{4} |f|^2 e_3 \right) + \lot \\
           &= \rhoF - \frac{1}{2} f^b \bbF_b + \frac{1}{2} \fb^b \bF_b - \frac{1}{2} \in_{ab} \fb^a f^b \dual\rhoF + \frac{1}{2} \rhoF (f \c \fb) + \lot,
  \end{align*}
  as claimed.
\end{proof}
\subsection{Deformation of surfaces in \texorpdfstring{$\RR$}{}}


\begin{definition}
  \label{definition:Deformations}
  We say that $\S$ is an \emph{$O(\dg)$ deformation of $ \ovS$} if there exist smooth scalar functions $U, S$ defined on $\ovS$ and a map 
  a map $\Psi:\ovS\to \S $ verifying, on any coordinate chart $(y^1, y^2) $ of $\ovS$, 
  \begin{equation*}
    \Psi(\ovu, \ovs, y^1, y^2)=\left( \ovu+ U(y^1, y^2 ), \, \ovs+S(y^1, y^2 ), y^1, y^2 \right)
  \end{equation*}
  with $(U, S)$ smooth functions on $\ovS$ of size $\dg$.
\end{definition}

\begin{definition}
  Given a deformation $\Psi:\ovS\to \S$ we say that 
  a new frame $(e_3', e_4', e_1', e_2')$ on $\S$, obtained from the standard frame $(e_3, e_4, e_1, e_2)$ via the transformation \eqref{eq:Generalframetransf}, is \emph{$\S$-adapted} if the horizontal vectorfields $e'_1, e'_2$ are tangent to $\S$ or, equivalently $e_3', e_4' $ are perpendicular to $S$.
\end{definition}

\begin{lemma}\label{lemma:comparison-gaS-ga}
  Let $\ovS \subset \RR$. Let $\Psi:\ovS\longrightarrow \S $ be a deformation generated by the functions $(U, S)$ as in Definition \ref{definition:Deformations} and denote by $g^{\S,\#}$ the pull back of the metric $g^\S$ to $\ovS$. Assume the bound, for $s\le s_{max}+1$,
  \begin{equation}\label{assumption-UV-dg}
    \| (U, S)\|_{L^\infty(\ovS)} +r ^{-1} \big\|(U, S)\big\|_{\hk_{s}(\ovS)} \les \dg.
  \end{equation}
  Then 
  \begin{equation}
    \frac{r^\S}{\ovr}= 1 + O(r ^{-1} \dg ),
  \end{equation}
  where $r^\S$ is the area radius of $\S$ and $\ovr$ that of $\ovS$. 
  
  Also,
  \begin{equation}
    \big\| g^{\S, \#} -\ovg\big\|_{L^\infty} +r^{-1} \big\| g^{\S, \#} -\ovg\big\|_{\hk_{s}(\ovS)}\les \dg r.
  \end{equation}
\end{lemma}

\begin{proof}
  See Lemma 5.8 in \cite{klainermanConstructionGCMSpheres2022}.
\end{proof}

\subsection{Adapted non canonical \texorpdfstring{$\ell=1$ modes}{}}


Consider a deformation $\Psi:\ovS \to \S$ and recall the existence of the family of scalar functions $\Jp$, $p\in\big\{0, +, -\big\}$, on $\RR$ introduced in assumption {\bf A4}, see Definition \ref{def:l=1}, which form a basis of the $\ell=1$ modes on the spheres $S(u,s)$ of $\RR$, and hence in particular on $\ovS$. We denote by 
$\Jpov $ the restriction of the family $\Jp$ to $\ovS$.

\begin{definition} 
  \label{def:ell=1sphharmonicsonS}
  We define the \emph{basis of adapted $\ell=1$ modes $\JpS$ on $\S$} by
  \begin{align*}
    \JpS = \Jpov\circ\Psi^{-1}, \qquad p\in\big\{ -, 0, +\big\}.
  \end{align*}
\end{definition}

\section{GCM system of adapted frame transformations and deformation following \texorpdfstring{\cite{klainermanConstructionGCMSpheres2022}}{}}\label{sec:GCM-system}

The following crucial system describes a $\S$-adapted frame transformation and deformation derived by Klainerman-Szeftel in \cite{klainermanConstructionGCMSpheres2022}.
\begin{proposition}[Corollary 4.6 and Proposition 5.14 in \cite{klainermanConstructionGCMSpheres2022}]
  Consider a fixed deformation $\Psi:\ovS\to\S$ from the background sphere $\ovS$ to a $O(\epg)$-sphere $\S$ in $\RR$, generated by functions $U, S:\ovS\to\RRR$. Let $F=(f,\fb,\ovla=\la-1)$ satisfy
  \[|F| \ll 1.\]
  Then a new frame $e_4^\S, e_3^\S, e_1^\S, e_2^\S$ on $\S$, generated by $(f,\fb,\la)$ from the reference frame $e_4, e_3, e_1, e_2$ via the transformation formulas \eqref{eq:Generalframetransf}, is $\S$-adapted if and only if the corresponding variables $(U, S, f, \fb, \ovla)$ solve the following coupled system:
  \begin{equation}\label{eq:qequivalentGCMsystemwhicisnotsolvabledirectly:1}
    \begin{split}
      \curl^\S f &= -\err_1[\curl^\S f], \\
      \curl^\S \fb &= -\err_1[\curl^\S \fb],
    \end{split}
  \end{equation}
  \begin{equation}\label{eq:qequivalentGCMsystemwhicisnotsolvabledirectly:2}
    \begin{split}
      \div^\S f + \trch \ovla -\frac{2}{(r^\S)^2}\ovb
      ={}& \trch^\S - \frac{2}{r^\S} - \left(\trch - \frac{2}{r}\right) - \err_1[\div^\S f] - \frac{2(r - r^\S)^2}{r(r^\S)^2},\\
      \div^\S\fb - \trchb \ovla + \frac{2}{(r^\S)^2}\ovb
      ={}& \trchb^\S + \frac{2}{r^\S} - \left(\trchb + \frac{2}{r}\right) - \err_1[\div^\S \fb] + \frac{2(r - r^\S)^2}{r(r^\S)^2},\\
      \Delta^\S\ovla + V\ovla
      ={}& \mu^\S - \mu - \left(\omb + \frac{1}{4} \trchb\right)(\trch^\S - \trch)\\
         &+ \left(\om + \frac{1}{4} \trch\right)(\trchb^\S - \trchb) + \err_2[\lap^\S\ovla],\\
      \Delta^\S\ovb
      ={}& \frac{1}{2} \div^\S\left(\fb - \Up f + \err_1[\Delta^\S \ovb]\right), \quad \ov{\ovb}^\S = \ov{r}^\S - r^\S,
    \end{split} 
  \end{equation}
  \begin{equation}\label{eq:qequivalentGCMsystemwhicisnotsolvabledirectly:3}
    \begin{split}
      \pr_{y^a} S&= \Big( \SS(f, \fb, \Ga)_bY^b_{(a)} \Big)^\#,\\
      \pr_{y^a} U&=\Big(\UU(f, \fb, \Ga)_bY^b_{(a)}\Big)^\#,
    \end{split}
  \end{equation}
  where $\ovb = r - r^\S$ and the potential $V$ is given by
  \[
    V \vcentcolon= -\left( \frac{1}{2} \trch \trchb + \trch \omb + \trchb \om \right).
  \]
  The error terms satisfy the schematic estimates
  \[
    \begin{split}
      r\err_1 &= F\c(r\Ga_b) + F\c(r\nab^\S)^{\le 1}F + r^{-1}F, \\
      r^2\err_2 &= (r\nab^\S)^{\le 1}(r\err_1) + (F + \Ga_b)\c r\dk \Ga_b,
    \end{split}
  \]
  with their exact expressions depending on the specific equation.
\end{proposition}

Note that the system
\eqref{eq:qequivalentGCMsystemwhicisnotsolvabledirectly:1}--\eqref{eq:qequivalentGCMsystemwhicisnotsolvabledirectly:3} is a priori
not solvable directly. To circumvent this, we instead solve a modified
system in which appropriate averages are added to enforce solvability,
following \cite{klainermanConstructionGCMSpheres2022}:

\begin{definition}
  We say that $F=(f,\fb,\ovla=\la-1)$ with $|F| \ll 1$ satisfies the \emph{modified solvable system} if
  \begin{equation}\label{eq:qequivalentGCMsystemwhicisnotsolvabledirectly:1:bis}
    \begin{split}
      \curl^\S f &= -\err_1[\curl^\S f] + \ov{\err_1[\curl^\S f]}^\S, \\
      \curl^\S \fb &= -\err_1[\curl^\S \fb] + \ov{\err_1[\curl^\S \fb]}^\S 
    \end{split}
  \end{equation}
  \begin{equation}\label{eq:qequivalentGCMsystemwhicisnotsolvabledirectly:2:bis}
    \begin{split}
      \div^\S f + \trch \ovla -\frac{2}{(r^\S)^2}\ovb ={}& \trch^\S - \frac{2}{r^\S} - \left(\trch - \frac{2}{r}\right) - \err_1[\div^\S f] - \frac{2(r - r^\S)^2}{r(r^\S)^2},\\
      \div^\S\fb - \trchb \ovla + \frac{2}{(r^\S)^2}\ovb ={}& \trchb^\S + \frac{2}{r^\S} - \left(\trchb + \frac{2}{r}\right) - \err_1[\div^\S \fb] + \frac{2(r - r^\S)^2}{r(r^\S)^2},\\
      \Delta^\S\ovla + V\ovla ={}& \mu^\S - \mu - \left(\omb + \frac{1}{4} \trchb\right)(\trch^\S - \trch) \\
                                                      &+ \left(\om + \frac{1}{4} \trch\right)(\trchb^\S - \trchb) + \err_2[\lap^\S\ovla],\\
      \Delta^\S\ovb ={}& \frac{1}{2} \div^\S\left(\fb - \Up f + \err_1[\Delta^\S \ovb]\right), \quad \ov{\ovb}^\S = \ov{r}^\S - r^\S,
    \end{split}
  \end{equation}
  \begin{equation}\label{systemUU-SS-derived}
    \begin{split}
      \lapzero U&=\divzero\left(\big(\UU(f, \fb, \Ga)\big)^\#\right),\\
      \lapzero S&=\divzero\left(\big(\SS(f, \fb, \Ga)\big)^\#\right).
    \end{split}
  \end{equation}
\end{definition}

\section{Mass-centered (charged) GCM spheres}

In this section, we construct mass-centered (charged) GCM spheres in
Theorem \ref{thm:GCM-spheres-J} and also the intrinsic GCM sphere in Theorem
\ref{thm:intrinsic-GCM}. 

\subsection{Definition of mass-centered (charged) GCM spheres}\label{section:GCMspheres}
We first define the generic (charged) GCM spheres which satisfy conditions mimicking those in \eqref{eq:condition-RN} in Reissner-Nordstr\"om. 

\begin{definition}[Generic (charged) GCM spheres]\label{def:GCM-original}
  A topological\footnote{The GCM spheres are denoted ${\S}$ in contrast with any sphere $S$ of the foliation in $\RR$.} sphere ${\S}$ in $\RR$, endowed with a null frame $(e_3^{\S}, e_4^{\S}, e_1^{\S}, e_2^{\S})$ adapted\footnote{This means that $e_1^{\S}, e_2^{\S}$ are tangent to $S$} to it, is called a \emph{GCM sphere with respect to a basis of $\ell=1$ modes $J^{({\S}, p)}$} (as in Definition \ref{def:l=1}) if
  \begin{align}\label{eq:conditions-GCM1}
    \trch^{\S} -\frac{2}{r^{\S}}=0, \qquad \Big(\trchb^{\S}+\frac{2\Up^{\S}}{r^{\S}}\Big)_{\ell\geq 2, J^{({\S})}}=0, \qquad \Big(\mu^{\S}- \frac{2M^{\S}}{(r^{\S})^3}+\frac{2\left((Q^\S)^2+(e^\S)^2\right)}{(r^{\S})^4}\Big)_{\ell\geq 2, J^{({\S})}}=0,
  \end{align}
  as well as
  \begin{equation}\label{eq:conditions-GCM1-e}
    e^\S\vcentcolon=\frac{1}{4\pi}\int_\S \dual\rhoF^\S =0,
  \end{equation}
  where $r^{\S}$, $M^{\S}$, $Q^{\S}$ and $e^{\S}$ denote the area radius, the mass, the electric charge and magnetic charge of ${\S}$ as in Definition \ref{def:parameters}.
\end{definition}
\begin{remark}\label{rmk:totalchargeconserved}
  The quantity $(Q^\S)^2 + (e^\S)^2$ in \eqref{eq:conditions-GCM1} represents the square of total electromagnetic charge of the sphere $\S$, defined in Definition \ref{def:parameters}. This combination is conserved under the natural $ U(1) $-action on the Maxwell equations, which acts on the complexified Maxwell 2-form $ \mathcal{F} \vcentcolon= \F + i \, {}^*\F $ by phase rotation:
  \[
    \mathcal{F} \mapsto e^{i\theta} \mathcal{F}, \qquad \theta \in \mathbb{R}.
  \]
  Under this action, the electric and magnetic components mix as:
  \[
    \rhoF \mapsto \cos\theta\, \rhoF + \sin\theta\, \dual\rhoF, \qquad 
    \dual\rhoF \mapsto \cos\theta\, \dual\rhoF - \sin\theta\, \rhoF.
  \]
  While the individual charges $ Q^\S $ and $ e^\S $ are not invariant, the square of total charge $ (Q^\S)^2 + (e^\S)^2 $ is preserved by this symmetry.
\end{remark}
An equivalent definition is the following:
\begin{definition}\label{definition:GCMS}

  We say that $\S\subset \RR$, endowed with an adapted
  frame\footnote{i.e.$ (e^\S_1, e_2 ^\S)$ are tangent to $\S$.}
  $(e_1^\S, e_2^\S, e_3^\S, e_4^\S)$, is a \emph{generic charged GCM
  sphere} if the following hold true:

  \begin{equation}
    \label{def:GCMC:00}
    \begin{split}
      \trch^\S&=\frac{2}{r^\S},\\
      \trchb^\S &=-\frac{2}{r^\S}\Up^\S+ \Cb^\S_0+\sum_p \CbpS \JpS,\\
      \mu^\S&= \frac{2m^\S}{(r^\S)^3}-\frac{(Q^\S)^2+(e^\S)^2}{(r^\S)^4} + M^\S_0+\sum _p\MpS \JpS,
    \end{split}
  \end{equation}
  and
  \begin{equation}\label{def:GCMC:e=0}
    \overline{\dual\rhoF^\S}^\S=0,
  \end{equation}
  for some constants $\Cb^\S_0,\, \CbpS, \, M^\S_0, \, \MpS, \, p\in\{-,0, +\}$.
  In addition, since the $\S$-frame is integrable, we also have
  \begin{equation}
    \label{Conditions:GCMS-automatic}
    \atrch^\S=\atrchb^\S=0.
  \end{equation}
\end{definition}

In this work, we rely on GCM spheres whose center of mass vanishes,
which we call mass-centered GCM spheres.

\begin{definition}[Mass-centered GCM spheres]\label{def:mass-centered-sphere} 

  A GCM sphere ${\S}$ as in Definition \ref{def:GCM-original} is called an (electrovacuum) \emph{mass-centered GCM sphere} with respect to a basis of $\ell=1$ modes $J^{({\S}, p)}$ if in addition to \eqref{eq:conditions-GCM1} and \eqref{eq:conditions-GCM1-e}, we also have

  \begin{align}
    \label{eq:condition-center-mass}
    \bm{C}^{\S}_{\ell=1,J^{({\S})}}=0.
  \end{align}
\end{definition}

\begin{remark}
  In \cite{klainermanEffectiveResultsUniformization2022}, the intrinsic GCM sphere $S_*$ is a mass-centered GCM sphere, for which $(\div\b)_{\ell=1}=0$ with respect to a canonical basis of $\ell=1$ modes, but the GCM hypersurface constructed in \cite{shenConstructionGCMHypersurfaces2023} and used in \cite{klainermanKerrStabilitySmall2023} $\Sigma_*$, is not necessarily foliated by spheres for which $(\div\b)_{\ell=1}=0$ with respect to a canonical basis of $\ell=1$ modes.
\end{remark}

\subsection{Existence of mass-centered (charged) GCM spheres}

In this section, we prove the existence of mass-centered (charged) GCM spheres in Theorem Theorem \ref{thm:GCM-spheres-J}.

\subsubsection{Existence of generic (charged) GCM spheres}
The following theorem is analogous to the main result of \cite{klainermanConstructionGCMSpheres2022}.
\begin{theorem}[Existence of generic (charged) GCM spheres]
  \label{Theorem:ExistenceGCMS1}\label{thm:GCM-spheres}
  Let $m_0, Q_0>0$ be constants. Let $0<\dg\leq \epg $ two sufficiently small constants, and let $(\ug, \sg, \rg)$ three real numbers with $\rg$ sufficiently large so that
  \begin{align*}
    \epg\ll m_0, Q_0\ll \rg.
  \end{align*}
  Let a fixed spacetime region $\RR$, as in Definition \ref{definition-spacetime-region-RR}, together with a $(u, s)$ outgoing geodesic foliation verifying the assumptions ${\bf A1-A4}$, see section \ref{subsubsect:regionRR2}. Let $\ovS=S(\ovu, \ovs)$ be a fixed sphere from this foliation, and let $\rg$ and $\mg$ denoting respectively its area radius and its Hawking mass. Assume that the GCM quantities $\trch, \trchb, \mu$ of the background foliation verify the following:
  \begin{equation*}
    \begin{split}
      \trch&=\frac{2}{r}+\dot{\trch},\\
      \trchb&=-\frac{2\Up}{r} + \Cb_0+\sum_p \Cbp \Jp+\dot{\trchb},\\
      \mu&= \frac{2m}{r^3} -\frac{Q^2}{r^4}+ M_0+\sum _p\Mp \Jp+\dot{\mu},
    \end{split}
  \end{equation*}
  where
  \begin{equation*}
    |\Cb_0, \Cbp| \les r^{-2} \epg, \qquad |M_0, \Mp| \les r^{-3} \epg,
  \end{equation*}
  and
  \begin{equation*}
    \big\| \dot{\trch}, \dot{\trchb}\|_{\hk_{s_{max} }(\S) }\les r^{-1}\dg,\qquad 
    \big\|\mudot\| _{\hk_{s_{max} }(\S) }\les r^{-2}\dg.
  \end{equation*}
  Then
  for any fixed pair of triplets $\La, \Lab \in \mathbb{R}^3$ verifying $|\La|,\, |\Lab| \les \dg$,
  there 
  exists a unique GCM sphere $\S=\S^{(\La, \Lab)}$, which is a deformation of $\ovS$, 
  such that the GCM conditions of Definition \ref{definition:GCMS} are verified,
  i.e. there exist constants $\Cb^\S_0,\, \CbpS$, \, $ M^\S_0$, \, $\MpS, \, p\in\{-,0, +\}$ for which
  \begin{equation*}
    \begin{split}
      \trch^\S&=\frac{2}{r^\S},\\
      \trchb^\S &=-\frac{2}{r^\S}\Up^\S+ \Cb^\S_0+\sum_p \CbpS \JpS,\\
      \mu^\S&= \frac{2m^\S}{(r^\S)^3} -\frac{(Q^\S)^2}{(r^\S)^4} + M^\S_0+\sum _p\MpS \JpS,\\
      \overline{\dual\rhoF^\S}^\S&=0,
    \end{split}
  \end{equation*}
  where we recall that $\JpS=\Jp\circ\Psi^{-1}$, see Definition \ref{def:ell=1sphharmonicsonS}. Moreover, 
  \begin{equation}\label{GCMS:l=1modesforffb}
    (\div^\S f)_{\ell=1,J^{(\S)}}=\La, \qquad (\div^\S\fb)_{\ell=1,J^{(\S)}}=\Lab,
  \end{equation}
  where we recall that the $\ell=1$ modes for scalars on $\S$ are defined by \eqref{eq:defell=1foroneformofepgspheres}.

  The resulting deformation has the following additional properties:
  \begin{enumerate}
  \item The triplet $(f,\fb,\ovla)$ verifies
    \begin{equation}\label{eq:ThmGCMS1}
      \|(f,\fb, \ovla)\|_{\hk_{s_{max}+1}(\S)} \les \dg. 
    \end{equation}
  \item The GCM constants $\Cb^\S_0,\, \CbpS$, \, $ M^\S_0$, \, $\MpS, \, p\in\{-,0, +\}$ verify
    \begin{equation*}
      \begin{split}
        \big| \Cb^\S_0-\Cb_0\big|+\big| \CbpS-\Cbp\big|&\les r^{-2}\dg,\\
        \big| M^\S_0-M_0\big|+\big| \MpS-\Mp\big|&\les r^{-3}\dg.
      \end{split}
    \end{equation*}

  \item The volume radius $r^\S$ verifies
    \begin{equation*}
      \left|\frac{r^\S}{\rg}-1\right|\les r^{-1} \dg.
    \end{equation*}
  \item The parameter functions $U, S$ of the deformation verify
    \begin{equation}\label{eq:ThmGCMS4}
      \|( U, S)\|_{\hk_{s_{max}+1}(\ovS)} \les r \dg.
    \end{equation}

  \item The Hawking mass $m^\S$ and charge $Q^\S$ of $\S$ verify the estimate
    \begin{equation*}
      \big|m^\S-\ovm\big|+ \big|Q^\S -\ovQ\big|\les \dg. 
    \end{equation*}
  \item The well defined\footnote{These include the Ricci coefficients $\trch^\S, \trchb^\S, \chih^\S, \chibh^\S, \ze^\S$, the curvature components $\a^\S, \b^\S, \rho^\S, \rhod^\S, \bb^\S, \aa^\S$, the electromagnetic components $\bF^\S, \rhoF^\S, \dual\rhoF^\S, \bbF^\S$, and the mass aspect function $\mu^\S$. In contrast, quantities such as $\eta^\S, \etab^\S, \xi^\S, \xib^\S, \om^\S, \omb^\S$ involve derivatives in the $e_3^\S$ and $e_4^\S$ directions and are not directly accessible on $\S$.}
    Ricci coefficients, curvature components and electromagnetic components of $\S$ verify,
    \begin{equation*}
      \begin{split}
        \| \Ga^\S_g\|_{\hk_{s_{max} }(\S) }&\les \epg r^{-1},\\
        \| \Ga^\S_b\|_{\hk_{s_{max} }(\S) }&\les \epg.
      \end{split}
    \end{equation*}
  \end{enumerate}
\end{theorem}
\begin{proof}[Proof of Theorem \ref{Theorem:ExistenceGCMS1}]
  We first aim to solve the modified solvable system \eqref{eq:qequivalentGCMsystemwhicisnotsolvabledirectly:1:bis}--\eqref{systemUU-SS-derived} for the frame transformation and deformation parameters, subject to the generic (charged) GCM conditions \eqref{def:GCMC:00}, the prescribed $\ell=1$ conditions \eqref{GCMS:l=1modesforffb}, as well as the normalization
  \begin{equation*}
    U(\text{South}) = S(\text{South}) = 0
  \end{equation*}
  at the South Pole of the background sphere $\ovS$. In view of the GCM conditions \eqref{def:GCMC:00}, modified solvable system \eqref{eq:qequivalentGCMsystemwhicisnotsolvabledirectly:1:bis}--\eqref{systemUU-SS-derived} can be written in the following form:
  \begin{equation}
    \begin{split}\label{eq:KS-GCM}
      \curl ^\S f &= h_1 -\ov{h_1}^\S,\\
      \curl^\S \fb&= \underline{h}_1 - \ov{\underline{h}_1}^\S,\\
      \div^\S f + \frac{2}{r^\S} \ovla -\frac{2}{(r^\S)^2}\ovb &=h_2,\\
      \div^\S\fb + \frac{2}{r^\S} \ovla +\frac{2}{(r^\S)^2}\ovb
                  &= \Cbdot_0+\sum_p \Cbpdot \JpS +\underline{h}_2,\\
      \left(\Delta^\S+\frac{2}{(r^\S)^2}\right)\ovla &= \Mdot_0+\sum _p\Mpdot \JpS+\frac{1}{2r^\S}\left(\Cbdot_0+\sum_p \Cbpdot \JpS\right) +h_3,\\
      \Delta^\S\ovb-\frac{1}{2}\div^\S\Big(\fb - f\Big) &= h_4 -\ov{h_4}^\S , \qquad \ov{\ovb}^\S=b_0,
    \end{split}
  \end{equation}
  where $\ovb\vcentcolon=r-r^\S$, and
  \begin{equation}
    \begin{split}\label{eq:hhhhh}
      &\left\|\nu^l(h_1,\underline{h}_1,h_2,\underline{h}_2,h_4)\right\|_{\hk_{s}(\S)}+r\left\|\nu^l(h_3)\right\|_{\hk_{s}(\S)}\\
      \les{}& r^{-1}\dg+(\epg r^{-1}+r^{-2})\left(\left\|(\nab^\S_{\nu})^{l}(F)\right\|_{\hk_{s}(\S)}+r^{-1}\big\|\nu^{l}(\ovb)\big\|_{\hk_{s}(\S)}\right)\\
      &+r^{-1}\left(\left\|(\nab^\S_{\nu})^{\leq l-1}(F)\right\|_{\hk_{s}(\S)}+r^{-1}\big\|\nu^{\leq l-1}\ovb\big\|_{\hk_{s}(\S)}\right)+r^{-2}\dg\|\nu^{\leq l-1}(b^\S)\|_{\hk_{s}(\S)}.
    \end{split}
  \end{equation} 
  See also Lemma 4.21 in \cite{shenConstructionGCMHypersurfaces2023} for the system \eqref{eq:KS-GCM}. The method of solving \eqref{eq:KS-GCM} follows the general strategy developed in \cite[Section 6]{klainermanConstructionGCMSpheres2022}. One first constructs an iterative scheme to solve the nonlinear system \eqref{eq:qequivalentGCMsystemwhicisnotsolvabledirectly:1:bis}--\eqref{systemUU-SS-derived}. Using suitable elliptic and transport estimates, the sequence is shown to converge to a limit $(U^{(\infty)}, S^{(\infty)}, f^{(\infty)}, \fb^{(\infty)}, \ovla^{(\infty)})$ solving the modified system. To conclude, one proves that the resulting frame induced by $(f^{(\infty)}, \fb^{(\infty)}, \ovla^{(\infty)})$ coincides with the frame adapted to the limiting deformed sphere $\S^{(\infty)}$. This implies that the limit also satisfies the original GCM system \eqref{eq:qequivalentGCMsystemwhicisnotsolvabledirectly:1}--\eqref{eq:qequivalentGCMsystemwhicisnotsolvabledirectly:3}, thus establishing the result. The remaining properties can be established in a similar way. 
  
  We now control the electric charge $Q^\S$:  
\begin{align*}
    Q^\S = \frac{1}{4\pi}\int_{\S}\rhoF^\S, 
    \qquad 
    \ovQ = \frac{1}{4\pi}\int_{\ovS}\rhoF.
\end{align*}
By Lemma \ref{lemma:comparison-gaS-ga}, see also Lemma~7.3 of \cite{klainermanConstructionGCMSpheres2022}, which provides control of the induced metric, we have
\begin{align*}
    \left|\int_\S \rhoF - \int_{\ovS} \rhoF \right| \les \dg.
\end{align*}
Hence
\begin{align*}
    |Q^\S - \ovQ| \les \left|\int_\S (\rhoF^\S - \rhoF)\right| + \dg.
\end{align*}
Using the transformation formula for $\rhoF$, which involves only $(f,\fb)$, together with the control of $(f,\fb)$ already obtained, we deduce
\begin{align*}
    |Q^\S - \ovQ| \les \dg,
\end{align*}
as required.

Finally, we use the $U(1)$ symmetry of the Einstein–Maxwell equations to eliminate the magnetic charge. By applying a constant phase rotation to the complexified field
\[
   \mathcal{F} = \F + i\, {}^*\F,
\]
we may assume that the magnetic charge of $\S$ vanishes, i.e.
\[
    e^\S = \frac{1}{4\pi} \int_\S \dual\rhoF^\S = 0,
\]
so that condition \eqref{def:GCMC:e=0} is satisfied. Moreover, by Remark~\ref{rmk:totalchargeconserved}, the quantity $(Q^\S)^2 + (e^\S)^2$ is conserved under such a phase rotation, and therefore the system \eqref{eq:KS-GCM} remains unchanged. This completes the proof of Theorem~\ref{thm:GCM-spheres}.
\end{proof}

\subsubsection{Differentiability with respect to the parameters \texorpdfstring{$(\La, \Lab)$}{}}


The following proposition investigates the differentiability with respect to $(\La, \Lab)$ of the various 
quantities appearing in Theorem \ref{Theorem:ExistenceGCMS1}.
\begin{proposition}
\label{prop:diff-wrt-La-Lab}
  Under the assumptions of Theorem \ref{Theorem:ExistenceGCMS1}, let
  $\S^{(\La, \Lab)}$ denote the deformed spheres constructed in Theorem
  \ref{Theorem:ExistenceGCMS1} for parameter
  $\La, \Lab \in \mathbb{R}^3$ verifying $|\La|,\, |\Lab| \les \dg$.
  Then
  \begin{enumerate}
  \item The transition parameters $(f, \fb, \ovla)$ are continuous and differentiable with respect to $\La, \Lab $ and verify
    \begin{equation*}
      \begin{split}
        \frac{\pr f }{\pr \La}&=O\big( r^{-1}\big), \quad \frac{ \pr f }{\pr \Lab}=O\big(\dg r^{-1} \big), \\ 
        \frac{\pr\fb }{\pr \La}&=O\big(\dg r^{-1}\big), \quad \frac{\pr\fb}{\pr \Lab}=O\big( r^{-1} \big),\\
        \frac{\pr\ovla }{\pr \La}&=O\big(\dg r^{-1}\big), \quad \frac{\pr\ovla }{\pr \Lab}=O\big(\dg r^{-1}\big).
      \end{split}
    \end{equation*}
  \item The parameter functions $U, S$ of the deformation are continuous and differentiable with respect to $\La, \Lab $ and verify
    \begin{equation*}
      \frac{\pr U}{\pr \La}= O(1), \qquad \frac{\pr U}{\pr \Lab}= O(1), \qquad \frac{\pr S}{\pr \La}= O(1), \qquad \frac{\pr S}{\pr \Lab}= O(\dg).
    \end{equation*}
  \item Relative to the coordinate system induced by $\Psi$, the metric $g^\S$ of $\S=\S^{\La, \Lab}$ is continuous with respect to the parameters $\La, \Lab$ and verifies 
    \begin{align*}
      \big\| \pr_\La g^\S, \, \pr_{\Lab} g^\S\|_{L^\infty(\S)} &\les O( r^2).
    \end{align*}
  \end{enumerate}
\end{proposition}

\begin{proof}
  The proof follows by differentiating the equations satisfied by $(f, \fb, \la)$ and $(U, S)$ with respect to $(\La, \Lab)$ and relying on the estimates derived for $(f, \fb, \la)$ and $(U, S)$ in Theorem \ref{Theorem:ExistenceGCMS1}. The details are cumbersome but straightforward, and left to the reader.
\end{proof}

\subsubsection{Center of mass function on a generic (charged) GCM sphere}
\begin{proposition}\label{prop:C-transform}
Under the same assumptions as Theorem \ref{Theorem:ExistenceGCMS1}, 
  on a generic (charged) GCM sphere $\S$ in Theorem \ref{Theorem:ExistenceGCMS1}, the center of mass function $\bm{C}^\S$ satisfies
  \begin{equation}
    \bm{C}^\S+\left(\frac{3m^\S}{(r^\S)^3}-\frac{(Q^\S)^2}{2(r^\S)^4} \right)\div^\S f= \bm{C} + \div^\S (r^{-1}\Gag\c F)+r^{-1}(\dk^{\leq 1} \Gag)\c F,
  \end{equation}
  where $F=(f,\fb,\ovla)$.
\end{proposition}
\begin{proof}
  From the frame transformation formula Proposition 3.3 in \cite{klainermanConstructionGCMSpheres2022} and Lemma \ref{lem:EM-transform}, we know
  \begin{align*}
    \b^\S&=\b +\frac{3}{2}f\rho + r^{-1}\Gag\c F,\\
    \bF^\S&=\bF+\rhoF f + \Gag\c F,\\
    \rhoF^\S&=\rhoF+\Gab\c F.
  \end{align*}
  Therefore,
  \begin{align*}
    \b^\S+\left(\frac{3m^\S}{(r^\S)^3}-\frac{3(Q^\S)^2}{2(r^\S)^4} \right)f&=\b+r^{-1}\Gag\c F,\\
    \bF^\S-\frac{Q^\S}{(r^\S)^2}f&=\bF + \Gag\c F,\\
    \rhoFc^\S&=\rhoF + \Gab\c F.
  \end{align*}
  So from Definition \eqref{eq:definition-bm-C}, we compute
  \begin{align*}
    &\quad\bm{C}^\S+\left(\frac{3m^\S}{(r^\S)^3}-\frac{(Q^\S)^2}{2(r^\S)^4} \right)\div^\S f\\
    &=\div^\S\b^\S-\frac{Q^\S}{(r^\S)^2}\div^\S \bF^\S +\frac{2Q^\S}{(r^\S)^3}\rhoFc^\S+\left(\frac{3m^\S}{(r^\S)^3}-\frac{(Q^\S)^2}{2(r^\S)^4} \right)\div^\S f\\
    &=\div^\S\left(\b^\S+\left(\frac{3m^\S}{(r^\S)^3}-\frac{3(Q^\S)^2}{2(r^\S)^4} \right)f\right)-\frac{Q^\S}{(r^\S)^2}\div^\S \left(\bF^\S-\frac{Q^\S}{(r^\S)^2}f\right)+\frac{2Q^\S}{(r^\S)^3}\rhoFc^\S\\
    &=\div^\S\left(\b+r^{-1}\Gag\c F\right)-\frac{Q^\S}{(r^\S)^2}\div^\S \left(\bF + \Gag\c F\right)+\frac{2Q^\S}{(r^\S)^3}(\rhoF+\Gab\c F)\\
    &=\bm{C}+\div^\S(r^{-1}\Gag\c F)+r^{-2}(\dk^{\leq 1}\Gag)\c F,
  \end{align*}
  as stated.
\end{proof}

\subsubsection{Existence of mass-centered (charged) GCM spheres}

\begin{theorem}\label{thm:GCM-spheres-J}

  Let $\mathring{S}=S(\mathring{u}, \mathring{s})$ be a fixed sphere of the spacetime region $\RR$ as in Definition \ref{definition-spacetime-region-RR}, and let $0<\dg\leq \mathring{\epsilon}$ be two sufficiently small constants.

  Suppose in addition that, the following quantities of the background foliation in $\RR$ satisfy the following:
  \begin{align*}
    \begin{split}
      \trch &= \frac{2}{r}+ \dot{\trch}, \\
      \trchb &=- \frac{2\Up}{r}+\underline{C}_0(u,s)+\sum_p \underline{C}^{(p)}(u,s)J^{(p)} +\dot{\trchb}, \\
      \mu &= \frac{2M}{r^3}+\frac{2Q^2}{r^4} +M_0(u,s) + \sum_{p} M^{(p)}(u,s)J^{(p)} + \dot{\mu},
    \end{split} 
  \end{align*}
  where the scalar functions $\underline{C}_0$, $\underline{C}^{(p)}$, $M_0$, $M^{(p)}$ on $\RR$ depend only on $(u,s)$ and where 
  \begin{align*}
    \sup_{\RR} |\dkb^{\leq s_{\max}}(\dot{\trch}, \dot{\trchb})| \les r^{-2} \dg, \qquad \sup_{\RR} |\dkb^{\leq s_{\max}}\dot{\mu}| \les r^{-3} \dg.
  \end{align*}
  Moreover, 
  \begin{align*}
    \sup_{\RR} |\dkb^{\leq s_{\max}} \bm{C}_{\ell=1, J}| \les r^{-4} \dg.
  \end{align*}
  Then, for any fixed $\underline{\Lambda} \in \mathbb{R}^3$ verifying $|\underline{\Lambda}| \les \dg$ and a choice of triplet of functions $\widetilde{J}^{(p)}$ satisfying
  \begin{align*}
    \sum_{p = 0, +, -} \| J^{(p)} - \widetilde{J}^{(p)} \|_{\mathfrak{h}_{s_{\max}}(\mathring{S})} \les r \dg,
  \end{align*}
  and 
  \begin{align*}
    \partial_u \widetilde{J}^{(p)}(u,s,y^1, y^2) =\partial_s \widetilde{J}^{(p)}(u,s,y^1, y^2)=0,
  \end{align*}
  there exists a unique sphere ${\S}={\S}(\underline{\Lambda}, \widetilde{J}^{(p)})$, together with a null frame $(e_3^{\S}, e_4^{\S}, e_1^{\S}, e_2^{\S})$, which is a deformation of $\mathring{S}$ and an (electrovacuum) mass-centered GCM sphere as in Definition \ref{def:mass-centered-sphere}, with respect to the basis of $\ell=1$ modes $\widetilde{J}^{(p)}$, and such that 
  \begin{align*}
    ( \div^{\S}\underline{f})_{\ell=1,\widetilde{J}}=\underline{\Lambda},
  \end{align*}
  where $(f, \underline{f}, \lambda)$ denote the transition coefficients of the transformation from the background frame of $\RR$ to the frame adapted to ${\S}$. 
\end{theorem}
\begin{proof}
  In view of the existence of generic (charged) GCM spheres $\S(\Lambda, \underline{\Lambda})$ established in Theorem \ref{Theorem:ExistenceGCMS1}, the proof of Theorem \ref{thm:GCM-spheres-J} reduces to solving for $\Lambda = \Lambda(\underline{\Lambda})$ such that the mass-centered condition
  \[
    \bm{C}^\S_{\ell=1,\widetilde{J}} = 0
  \]
  is satisfied on the sphere $\S(\Lambda(\underline{\Lambda}), \underline{\Lambda})$, with respect to the given triplet of function $\widetilde{J}^{(p)}$.

  According to Proposition \ref{prop:C-transform}, \textbf{A1}, and Proposition \ref{prop:diff-wrt-La-Lab}, the parameter $\Lambda$ satisfies the nonlinear equation
  \begin{equation}\label{eq:definition-Lambda}
    \Lambda = \left( \frac{3m^\S}{(r^\S)^3} - \frac{(Q^\S)^2}{2(r^\S)^4} \right)^{-1} \left[ -\bm{C}^\S_{\ell=1,\widetilde{J}} + \bm{C}_{\ell=1,\widetilde{J}} \right] + H(\Lambda, \underline{\Lambda}),
  \end{equation}
  where $H$ is a continuously differentiable function of $(\Lambda, \underline{\Lambda})$ satisfying the estimates
  \begin{equation}\label{eq:Lip-H}
    |H| \les \mathring{\epsilon} \dg, \qquad |\partial_\Lambda H| + |\partial_{\underline{\Lambda}} H| \les \mathring{\epsilon}. 
  \end{equation}
   Therefore, solving the condition $\bm{C}^\S_{\ell=1,\widetilde{J}} = 0$ reduces to solving the fixed-point equation
  \[
    \Lambda = G(\Lambda, \underline{\Lambda}),
  \]
  where $G$ is a smooth function with small derivative in $\Lambda$. Since the coefficient 
  \[
    \frac{3m^\S}{(r^\S)^3} - \frac{(Q^\S)^2}{2(r^\S)^4}
  \]
  is uniformly bounded away from zero for sufficiently small $\delta_1$, the implicit function theorem guarantees the existence and uniqueness of a solution $\Lambda = \Lambda(\underline{\Lambda})$, at least for $\underline{\Lambda}$ in a neighborhood of size $\dg$. Observe that by triangle inequality, we deduce from \eqref{eq:definition-Lambda} that 
  \begin{align}\label{eq:estimate-Lambda}
    \Lambda \les r \epg .
  \end{align}
  This concludes the proof.
\end{proof}

We also recall the following proposition. 
\begin{proposition}[Proposition 4.13 in \cite{klainermanConstructionGCMSpheres2022}]\label{prop-a-priori-estimates-f-fb}
  Assume ${\S}$ is a given $O(\epg)$-sphere in $\RR$. Assume given a solution $(f, \fb, \ovla, \Cbdot_0, \Mdot_0, \Cbpdot, \Mpdot, \ovb)$ of the system \eqref{eq:KS-GCM} and verifying \eqref{GCMS:l=1modesforffb}. Then, the following a priori estimates are verified, for $3\leq s \leq s_{\max}+1$,
  \begin{align*}
    &  \| (f, \fb, \ovla - \ov{\ovla}^{\S}) \|_{\hk_{s}(\S)} + \sum_p \big(r^2 |\Cbpdot|+r^3 |\Mpdot| \big)+r^2|\Cbdot_0| + r^3 |\Mdot_0| + r |\ov{\ovla}^{\S}| \\
    \les{}& r \|(h_1 - \ov{h_1}^{\S}, \underline{h}_1 - \ov{\underline{h}_1}^\S, h_2 - \ov{h_2}^{\S}, \underline{h}_2 - \ov{\underline{h}_2}^\S)\|_{\hk_{s-1}(\S)}+ r^2 \| h_3 - \ov{h_3}^{\S}\|_{\hk_{s-2}}+ r \| h_4 - \ov{h_4}^{\S}\|_{\hk_{s-3}}\\
    &+|\Lambda| + |\underline{\Lambda}|+ |b_0|.
  \end{align*} 
\end{proposition}

\subsection{Existence of intrinsic GCM sphere}

\begin{theorem}\label{thm:intrinsic-GCM}
  Assume in addition the stronger assumption that for $k\le s_{max}$ and for a small enough constant $\delta_1>0$ with $\delta_1\ge \mathring{\epsilon}$,
      \begin{equation}
          \label{eq:A1-Strong}
          \norm*{\Gamma_g, \Gamma_b}_{k,\infty}\lesssim \delta_1 r^{-2},\qquad\norm*{\nabla_3\Gamma_g}_{k, \infty}\lesssim \delta_1 r^{-3},
      \end{equation}
      and the existence of a smooth family of scalar functions $J^{(p)}:\mathcal{R}\to \mathbb{R}$ for $p=0, +, -0$ verifying properties (7.2), (7.3) in \cite{klainermanEffectiveResultsUniformization2022} and $\partial_s J^{(p)}=\partial_u J^{(p)}=0$. Then we have that  
  on the background foliation of $\RR$ as in Theorem 7.3 in \cite{klainermanEffectiveResultsUniformization2022}, there exists a unique GCM deformation $\S$ of $\mathring{S}$, up to a rotation of $\mathbb{S}^2$, verifying:
  \begin{enumerate}
  \item the mass-centered GCM condition
    \begin{equation}\label{eq:bmC-ell1-can}
      \bm{C}^\S_{\ell=1,J^{(p,\S)}} = 0
    \end{equation}
    with respect to its canonical $\ell=1$ mode $J^{(p,\S)}$ as in Definition \ref{def:can-ell1-modes}, and
  \item the condition
    \begin{equation}\label{eq:trchbc-ell1-can}
      \left(\trchb^{\S} + \frac{2}{r^{\S}} \left(1 - \frac{2M^{\S}}{r^{\S}} + \frac{(Q^{\S})^2}{(r^{\S})^2} \right) \right)_{\ell=1,J^{(p,\S)}} = 0,
    \end{equation}
    again with respect to its canonical $\ell=1$ mode $J^{(p,\S)}$.
  \end{enumerate}
\end{theorem}

\begin{proof}
  Given $\Lambda, \underline{\Lambda} \in \mathbb{R}^3$, we apply Theorem 6.5 in \cite{klainermanEffectiveResultsUniformization2022} to construct a unique GCM sphere $\S = \S(\Lambda, \underline{\Lambda})$ endowed with canonical $\ell=1$ modes $J^{(p,\S)}$ as in Definition \ref{def:can-ell1-modes} from the effective uniformization Theorem \ref{thm:eff-unif}. Note that the magnetic charge $e^{\S(\Lambda, \underline{\Lambda})}$ can be fixed to zero by an additional $U(1)$ rotation of the electromagnetic 2-form $\F$.

  Similar to Lemma 7.5 of \cite{klainermanEffectiveResultsUniformization2022}, we have the identities:
  \begin{equation}\label{eq:intrinsicGCM-LaLab}
    \begin{split}
      \Lambda &= \left( \frac{3m^\S}{(r^\S)^3} - \frac{(Q^\S)^2}{2(r^\S)^4} \right)^{-1} \left[ -\bm{C}^\S_{\ell=1,J^{(p,\S)}} + \bm{C}_{\ell=1,J^{(p,\S)}} \right] + H_1[\Lambda, \underline{\Lambda}], \\
      \underline{\Lambda} &= \Up^{\S} \Lambda + \frac{r^{\S}}{3m^{\S}} \left[ \trchbc^{\S}_{\ell=1,J^{(p,\S)}} + \Up^{\S} \trchc^{\S}_{\ell=1,J^{(p,\S)}} - \trchbc_{\ell=1,J^{(p,\S)}} \right] + H_2[\Lambda, \underline{\Lambda}],
    \end{split}
  \end{equation}
  where $H_1, H_2$ are continuously differentiable functions of $(\Lambda, \underline{\Lambda})$ satisfying the estimates
  \begin{equation}\label{eq:est-H1H2}
    |H_1, H_2| \les \delta_1 \dg, \qquad |\partial_{\Lambda, \underline{\Lambda}}(H_1, H_2)| \les \delta_1.
  \end{equation}

  The existence of an intrinsic GCM sphere then follows by solving \eqref{eq:intrinsicGCM-LaLab} for $(\Lambda, \underline{\Lambda})$ such that conditions \eqref{eq:bmC-ell1-can} and \eqref{eq:trchbc-ell1-can} are satisfied. This is achieved via a fixed point argument, using the smallness estimates \eqref{eq:est-H1H2}. Uniqueness follows by the same method.
\end{proof}

\begin{remark}
  The effective uniformization theorem in
  \cite{klainermanEffectiveResultsUniformization2022} yields
  uniqueness only up to a rotation of the sphere. This residual
  freedom can be used to fix the virtual axis of rotation of the
  perturbed spacetime, as in
  \cite{klainermanConstructionGCMSpheres2022}. Introducing the
  \emph{angular momentum function}
\begin{align}\label{eq:def-J}
    \bm{J}\vcentcolon=\curl \b -\frac{Q}{r^2}\curl\bF +\frac{2Q}{r^3}\dual\rhoF,
\end{align}
we can exploit this freedom to show that one can construct an
intrinsic GCM sphere, as in Theorem \ref{thm:intrinsic-GCM}, which in
addition satisfies
    \begin{align}\label{eq:S_*-GCM}
      \int_{{\bf S}}J^{(+)}\bm{J}=0, \qquad \int_{{\bf S}}J^{(-)}\bm{J}=0.
    \end{align}
\end{remark}

\section{Existence of mass-centered GCM hypersurface}

We want to prove the following theorem, which relies on the construction of GCM spheres in Theorem \ref{thm:GCM-spheres-J}. In particular, we denote ${\S}_0$ a mass-centered GCM sphere obtained as a result of Theorem \ref{thm:GCM-spheres-J}.

\begin{theorem}\label{thm:GCMH}
  Let ${\S}_0$ be a fixed sphere of the spacetime region $\RR$ as in Definition \ref{definition-spacetime-region-RR}, and let $0<\dg\leq \mathring{\epsilon}$ be two sufficiently small constants.
  Let
  $\underline{\Lambda}_0 \in \mathbb{R}^3$ be such that $|\underline{\Lambda}_0|\les \dg$, and let $J^{({\S}_0, p)}$ be a basis of $\ell=1$ modes on ${\S}_0$ satisfying 
  \begin{align*}
    \sum_{p = 0, +, -} \| J^{(p)} - J^{({\S}_0, p)} \|_{\mathfrak{h}_{s_{\max}+1}(\mathring{S})} \les r \dg,
  \end{align*}
  and such that ${\S}_0$ is an (electrovacuum) mass-centered GCM sphere with respect to $J^{({\S}_0, p)}$ and such that 
  \begin{align*}
    ( \div^{\S_0}\underline{f})_{\ell=1,J^{({\S}_0)}}=\underline{\Lambda},
  \end{align*}
  where $(f, \underline{f}, \lambda)$ denote the transition coefficients of the transformation from the background frame of $\RR$ to the frame adapted to ${\S}_0$.

  Suppose in addition that, the following quantities of the background foliation in $\RR$ satisfy the following:
  \begin{align*}
    \begin{split}
      \trch &= \frac{2}{r}+ \dot{\trch}, \\
      \trchb &=- \frac{2\Up}{r}+\underline{C}_0(u,s)+\sum_p \underline{C}^{(p)}(u,s)J^{(p)} +\dot{\trchb}, \\
      \mu &= \frac{2M}{r^3}+\frac{2Q^2}{r^4} +M_0(u,s) + \sum_{p} M^{(p)}(u,s)J^{(p)} + \dot{\mu} 
    \end{split} 
  \end{align*}
  where the scalar functions $\underline{C}_0$, $\underline{C}^{(p)}$, $M_0$, $M^{(p)}$ on $\RR$ depend only on $(u,s)$ and where 
  \begin{align}\label{eq:control-trch-trchb-mu}
    \sup_{\RR} |\widetilde{\dk}^{\leq s_{\max}}(\dot{\trch}, \dot{\trchb})| \les r^{-2} \dg, \qquad \sup_{\RR} |\widetilde{\dk}^{\leq s_{\max}}\dot{\mu}| \les r^{-3} \dg, 
  \end{align}
  where $\widetilde{\dk}\vcentcolon= (e_3-(z+\underline{\Omega})e_4, \dkb)$ and $\dkb=r \nab$.
  Moreover, 
  \begin{align*}
    \sup_{\RR} |\widetilde{\dk}^{\leq s_{\max}} \bm{C}_{\ell=1, J}| \les r^{-4} \dg, \qquad |(\div\xib)_{\ell=1, J}| \les \dg.
  \end{align*}
  Suppose in addition that in $\RR$, $e_3(J^{(p)})$, $r-s$ and $e_3(r)-e_3(s)$ are small with respect to the parameter $\dg$.

  Then there exists a unique, local, smooth, spacelike hypersurface $\Sigma_0$ passing through ${\S}_0$, a scalar function $u^{\S}$ defined on $\Sigma_0$, whose level surfaces are topological spheres denoted by ${\S}$, a smooth collection of constant triplets $\underline{\Lambda}^{\S}$ and a triplet of functions $J^{({\S}, p)}$ defined on $\Sigma_0$ verifying 
  \begin{align*}
    \underline{\Lambda}^{{\S_0}}=\underline{\Lambda}_0, \qquad J^{({\S}, p)}|_{{\S}_0}=J^{({\S}_0, p)},
  \end{align*}
  such that the following conditions are verified:
  \begin{enumerate}
  \item The surfaces ${\S}$ of constant $u^{\S}$, together with a null frame $(e_3^{\S}, e_4^{\S}, e_1^{\S}, e_2^{\S})$, are (electrovacuum) mass-centered GCM spheres with respect to the basis of $\ell=1$ modes $J^{({\S}, p)}$.
  \item We have, for some constant $c_0$, 
    \begin{align*}
      u^{\S}+r^{\S}=c_0, \qquad \text{along $\Sigma_0$.}
    \end{align*}
  \item Let $\nu^{\S}$ be the unique vectorfield tangent to the hypersurface $\Sigma_0$, normal to ${\S}$, and normalized by $\g(\nu^{\S}, e_4^{\S})=-2$. Let $b^{\S}$ be the unique scalar function of $\Sigma_0$ such that $\nu^{\S}$ is given by 
    \begin{align*}
      \nu^{\S}=e_3^{\S}+b^{\S}e_4^{\S}.
    \end{align*}
    Then the following normalization condition holds true:
    \begin{align}\label{eq:condition-overline-b}
      \overline{b^{\S}}=-1-\frac{2M_{(0)}}{r^{\S}}+\frac{Q_{(0)}^2}{(r^{\S})^2}
    \end{align}
    where $\overline{b^{\S}}$ is the average value of $b^{\S}$ over ${\S}$ and $M_{(0)}$, $Q_{(0)}$ are constants.
  \item The following transversality conditions 
    \begin{align}\label{eq:transversality-condition}
      \xi^{\S}=0, \qquad \om^{\S}=0, \qquad \etab^{\S}=-\ze^{\S},
    \end{align}
    and 
    \begin{align*}
      e_4^{\S}(u^{\S})=0, \qquad e_4^{\S}(r^{\S})=1
    \end{align*}
    are assumed on $\Sigma_0$.

  \item We have the following identity on $\Sigma_0$:
    \begin{align}\label{eq:condition-xib-eta}
      (\div^{\S}\xib^{\S})_{\ell=1, J^{({\S})}}=0.
    \end{align}
  \item The transition coefficients $(f, \underline{f}, \mathring{\lambda})$ from the background foliation to that of $\Sigma_0$ verify
    \begin{align}\label{eq:bounds-f-fb-la}
      \| (f, \underline{f}, \mathring{\lambda})\|_{\mathfrak{h}_{s_{\max}+1}(\S)}+ \| \dk(f, \underline{f}, \mathring{\lambda})\|_{\mathfrak{h}_{s_{\max}}(\S)}\les \dg.
    \end{align}
  \end{enumerate}
\end{theorem}

\begin{remark}
  The value $\overline{b^{\S}}$ is free and should be prescribed. The choice \eqref{eq:condition-overline-b} coincides with the value for the hypersurface $\{u+r=c_0\}$ in Reissner-Nordstr\"om spacetime. 
\end{remark}

\begin{remark}
    In the context of nonlinear stability of Reissner-Nordstr\"om or Kerr-Newman, we will need to verify the existence of a mass-centered GCM hypersurface that starts from an intrinsic GCM sphere $S_*$ with a prescribed position (referred to as a dominance condition in \cite{fangEinsteinMaxwellEquationsMassCentered2025}). 
\end{remark}

\subsection{Proof of Theorem \texorpdfstring{\ref{thm:GCMH}}{}}

The proof consists in constructing $\Sigma_0$ by concatenating a family of mass-centered GCM spheres ${\S}(\underline{\Lambda}, \widetilde{J}^{(p)})$ as in Theorem \ref{thm:GCM-spheres-J} emanating from ${\S}_0$, and therefore Point 1 in Theorem \ref{thm:GCMH} will be automatically satisfied. On the other hand, \eqref{eq:condition-xib-eta} will be enforced thanks to a special choice of $\underline{\Lambda}^{\S}$.

More precisely, the hypersurface is described by ${\S}(\underline{\Lambda}(s), \widetilde{J}^{(p)}(s))$ where the $\ell=1$ basis $\widetilde{J}^{(p)}(s)$ is defined by exactly transporting it through $J^{(p)}{}'(s)=0$ starting with the given one at ${\S}_0$.

We divide the proof in four steps.

\subsubsection{Step 1: Apply Theorem \texorpdfstring{\ref{thm:GCM-spheres-J}}{} to spheres \texorpdfstring{in $\RR$}{}}

We assume $\RR$ is a spacetime region as in Definition \ref{definition-spacetime-region-RR}.

Let $\mathring{S}=S(\mathring{u}, \mathring{s})$ be a fixed sphere in $\RR$ with $\ell=1$ basis $J^{(p)}$ and let $\underline{\Lambda}_0 \in \mathbb{R}^3$ and $J_0^{(p)}$ be fixed such that $|\underline{\Lambda}_0| \les \dg$ and 
\begin{align*}
  \sum_{p = 0, +, -} \| J^{(p)} - J_0^{(p)} \|_{\mathfrak{h}_{s_{\max}}(S(u,s))} \les r \dg.
\end{align*}
Then, we denote ${\S}_0\doteq{\S}[\mathring{u},\mathring{s}, \underline{\Lambda}_0,J_0]$ the unique mass-centered GCM sphere which is a deformation of $\mathring{S}$ as in Theorem \ref{thm:GCM-spheres-J}.

Let $\Psi(s)$ be a real valued function and $\underline{\Lambda}(s)$ be a triplet of functions satisfying for $s \in \mathring{I}\vcentcolon=[\mathring{s}, s_1]$, with $|s_1-\mathring{s}|\les \epg$,
\begin{align*}
  |\underline{\Lambda}(s)|, \quad |\underline{\Lambda}'(s)|, \quad |\Psi(s) + s -c_0|, \quad |\Psi'+1| \les \dg, 
\end{align*}
and 
\begin{align*}
  \Psi(\mathring{s})=\mathring{u}, \qquad \underline{\Lambda}(\mathring{s})=\underline{\Lambda}_0,
\end{align*}
where $c_0=\mathring{u}+\mathring{s}$. We denote $$\psi(s)\vcentcolon= \Psi(s) + s -c_0.$$

We first define in a neighborhood of ${\S}_0$ a more general family of hypersurface, denoted $ \Sigma_{\#}$, within which we will make a choice of $\Sigma_0$. We define
\begin{align*}
  \Sigma_{\#}\vcentcolon= \bigcup_{s \in \mathring{I}} S(\Psi(s), s) = \{ u = \Psi(s), s \in \mathring{I} \}.
\end{align*}
Let $\widetilde{J}^{(p)}$ be a triplet of functions on $\Sigma_{\#}$ satisfying 
\begin{align}\label{eq:condition-J-intermediate-widetilde-J-thm}
  \sup_{s \in \mathring{I}}\sum_{p = 0, +, -} \| J^{(p)}(s) - \widetilde{J}^{(p)}(s)\|_{\mathfrak{h}_{s_{\max}}(S(\Psi(s), s))} \les r \dg,
\end{align}
and $\widetilde{J}^{(p)}(\mathring{s})=J_0^{(p)}$, where we denote $\widetilde{J}(s)\vcentcolon= \widetilde{J}|_{S(\Psi(s), s)}$. For every $s \in \mathring{I}$ we apply Theorem \ref{thm:GCM-spheres-J} to the background sphere $S(\Psi(s), s)$, the triplet $\underline{\Lambda}(s)$ and the triplet $\widetilde{J}(s)$ to obtain the (electrovacuum) mass-centered GCM sphere ${\S}[\Psi(s),s, \underline{\Lambda}(s),\widetilde{J}(s)]$ as a deformation of $S(\Psi(s),s)$. 

We define
\begin{align*}
  \Sigma_{\widetilde{J}}\vcentcolon= \bigcup_{s \in \mathring{I}} {\S}[\Psi(s),s, \underline{\Lambda}(s),\widetilde{J}(s)].
\end{align*}

\subsubsection{Step 2: Construct a family of \texorpdfstring{$\ell=1$}{} modes \texorpdfstring{$\widetilde{J}$}{} by propagating the one on the first sphere}

We construct, relying on a Banach fixed point argument, a family of basis of $\ell=1$ modes $\widetilde{J}(s)$ on $\Sigma_{\widetilde{J}}$. Morally speaking, we want to impose that 
\begin{align*}
  \nu^{\widetilde{J}}(\widetilde{J}(s))=0 \quad \text{on $\Sigma_{\widetilde{J}}$}, \qquad \widetilde{J}^{(p)}(\mathring{s})=J_0^{(p)},
\end{align*}
where $\nu^{\widetilde{J}}$ is the unique vectorfield tangent to $\Sigma_{\widetilde{J}}$ with $\g(\nu^{\widetilde{J}}, e_4^{\widetilde{J}})=-2$ and normal to ${\S}[\Psi(s),s, \underline{\Lambda}(s),\widetilde{J}(s)]$, where $(e_3^{\widetilde{J}}, e_4^{\widetilde{J}}, e_1^{\widetilde{J}}, e_2^{\widetilde{J}})$ denotes the adapted frame of ${\S}[\Psi(s),s, \underline{\Lambda}(s),\widetilde{J}(s)]$. Observe that we actually need to impose these conditions through the pull back of the deformation maps 
\begin{align*}
  \Phi^{\widetilde{J}}(s) : S(\Psi(s), s) \to {\S}[\Psi(s),s, \underline{\Lambda}(s),\widetilde{J}(s)] \subset \RR
\end{align*}
used in the proof of Theorem \ref{thm:GCM-spheres-J} in Step 1.

\begin{lemma}[Theorem 4.10 in \cite{shenConstructionGCMHypersurfaces2023}]\label{lemma:construction-Jtilde} There exists a unique $\widetilde{J}$ satisfying \eqref{eq:condition-J-intermediate-widetilde-J-thm} and verifying 
  \begin{align}\label{eq:conditions-Jtilde}
    \begin{split}
      \nu^{\widetilde{J}}\Big( \big((\Phi^{\widetilde{J}})^{-1} \big)^{\#} \widetilde{J} \Big)&=0 \qquad \qquad \text{on $\Sigma_{\widetilde{J}}$}, \\
      \big((\Phi^{\widetilde{J}})^{-1} \big)^{\#} \widetilde{J}^{(p)} &= J_0^{(p)} \qquad \qquad \text{on ${\S}_0$.} 
    \end{split}
  \end{align}
\end{lemma}
\begin{proof}
  See the proof of Theorem 4.10 in \cite{shenConstructionGCMHypersurfaces2023}. The proof is a Banach fixed-point argument through a contraction map $T$ which sends a triplet of functions $\widetilde{J}$ on $\Sigma_{\#}$ to another triplet of functions $T(\widetilde{J})$ on $\Sigma_{\#}$. The fixed point of $T$ satisfies \eqref{eq:conditions-Jtilde}.
  Notice that the proof of this lemma is independent of the GCM choice on $\Sigma$.
\end{proof}

By Lemma \ref{lemma:construction-Jtilde}, we define a triplet of functions $\widetilde{J}^{(p)}$ and a hypersurface
\begin{align}\label{eq:definition-Sigma}
  \Sigma \vcentcolon= \bigcup_{s \in \mathring{I}} {\S}[\Psi(s),s, \underline{\Lambda}(s),\widetilde{J}^{(p)}]= \bigcup_{s \in \mathring{I}} {\S}[\Psi(s),s, \underline{\Lambda}(s)].
\end{align}
\begin{remark}
  Notice that the hypersurface $\Sigma$ defined above depends on $\Psi$ and $\underline{\Lambda}$ which have not been fixed yet. They will be chosen to satisfy a system of ODEs.
\end{remark}

\begin{lemma}[Proposition 4.13 in \cite{shenConstructionGCMHypersurfaces2023}] The hypersurface $\Sigma$ constructed above is a smooth hypersurface.
\end{lemma}

According to the above construction, there is an ${\S}$-adapted null frame $(e_3^{\S}, e_4^{\S}, e_1^{\S}, e_2^{\S})$ on every mass-centered GCM sphere ${\S}[\Psi(s),s, \underline{\Lambda}(s)]$. We also assume the transversality conditions:
\begin{align*}
  \xi^{\S}=0, \qquad \om^{\S}=0, \qquad \etab^{\S}=-\ze^{\S}, \qquad e_4^{\S}(u^{\S})=0, \qquad e_4^{\S}(r^{\S})=1 \qquad \text{on $\Sigma$}
\end{align*}
where $u^{\S}$ is defined as $u^{\S}\vcentcolon= c_0 - r^{\S}$ on $\Sigma$ and $u^{\S}|_{{\S}_0}=\mathring{u}$. We also define 
\begin{align*}
  z^{\S}=e_3^{\S}(u^{\S}), \qquad \underline{\Omega}^{\S}=e_3^{\S}(r^{\S}),
\end{align*}
and the renormalized quantities 
\begin{align*}
  \widecheck{z}^{\S}\vcentcolon= z^{\S}-2, \qquad \widecheck{\underline{\Omega}}^{\S}\vcentcolon= \underline{\Omega}^{\S}+\Up^{\S}.
\end{align*}

We denote 
\begin{align*}
  J^{({\S}, p)}\vcentcolon= \big((\Phi^{\widetilde{J}})^{-1} \big)^{\#} \widetilde{J}^{(p)}
\end{align*}
the family of triplets on $\Sigma$. In particular, \eqref{eq:conditions-Jtilde} implies 
\begin{align}\label{eq:transport-J}
  \nu^{\S}( J^{({\S}, p)}) =0, \qquad J^{({\S}, p)}|_{{\S}_0}=J_0^{(p)},
\end{align}
where $\nu^{\S}$ is the unique vectorfield tangent to $\Sigma$ with $\g(\nu^{\S}, e_4^{\S})=-2$ and normal to ${\S}$.

\subsubsection{Step 3: Estimate the frame transformation coefficients}

Let $F=(f, \fb, \ovla)$ be the transition parameters of the frame transformation from the background frame $(e_3, e_4, e_1, e_2)$ to the adapted frame $(e_3^{\S}, e_4^{\S}, e_1^{\S}, e_2^{\S})$.

\begin{proposition}\label{prop:estimates-frame-coefficients}
  We have the following estimate:
  \begin{align}\label{eq:control-F}
    \| F\|_{\mathfrak{h}_{s_{\max}+1}({\S})} + \| \nab^{\S}_{\nu^{\S}} F\|_{\mathfrak{h}_{s_{\max}}({\S})} \les \mathring{\de}+ \Big|(\div^{\S}\xib^{\S})_{\ell=1, J^{({\S})}}\Big|+\Big| \overline{b^{\S}}+1+\frac{2M^{\S}}{r^{\S}}-\frac{(Q^{\S})^2}{(r^{\S})^2}\Big|.
  \end{align}
\end{proposition}
\begin{proof}
  As a consequence of the transformation formulas for $\eta, \etab$, $\xi, \xib$, and $\om, \omb$ in Lemma \ref{lem:EM-transform} and the transversality condition \eqref{eq:transversality-condition}, we deduce that $f, \fb, \lambda$ satisfy:
  \begin{align}\label{eq:nabla-nu-F}
    \begin{split}
      \nab^{\S}_{\nu^{\S}} (f) &= 2 (\eta^{\S}- \eta) -\frac 1 2 \trch (\fb + b^{\S} f ) + 2 f \omb + F \c \Gab + \lot, \\
      \nab^{\S}_{\nu^{\S}} (\fb) &=2(\xib^{\S}-\xib) -\frac 1 2 \fb \trchb -2\omb( \fb -b^{\S} f) + b^{\S}\big(2\nab^{\S}\lambda -\frac 1 2 \fb \trch \big)+F \c \Gab +\lot,\\
      \nab^{\S}_{\nu^{\S}} (\lambda) &=2(\omb^{\S}-\omb)-2\omb \ovla +F \c \Gab + \lot,
    \end{split}
  \end{align}
  where here $\lot$ denotes terms which are linear in $\Gag$, $\Gab$ and linear and higher order in $F$.
  Using that (see Lemma 4.16 in \cite{shenConstructionGCMHypersurfaces2023}), 
  \begin{align*}
    e_a^{\S}(z^{\S})&=(\ze^{\S}_a-\eta^{\S}_a) z^{\S}, \\
    e_a^{\S}(\underline{\Omega}^{\S}) &=(\ze_a^{\S}-\eta_a^{\S})\underline{\Omega}^{\S} -\xib_a^{\S}, \\
    b^{\S}&=-z^{\S}-\underline{\Omega}^{\S},
  \end{align*}
  and combining with the identities in \eqref{eq:nabla-nu-F} we infer that $(z^{\S}-\overline{z^{\S}})-(z-\overline{z})$ and $(\underline{\Omega}^{\S}-\overline{\underline{\Omega}^{\S}})-(\underline{\Omega}-\overline{\underline{\Omega}})$ can be estimated by $\eta^{\S}- \eta$ and $\xib^{\S}-\xib$, and therefore by $\nab^{\S}_{\nu^{\S}} (f, \fb)$. Consequently, from $b^{\S}=-z^{\S}-\underline{\Omega}^{\S}$, the same bound holds true for $(b^{\S}-\overline{b^{\S}})-(b-\overline{b})$.

  Next we use use that, for any scalar function $h$ on $\RR$ and any $1\leq l \leq s$, see Lemma 4.19 in \cite{shenConstructionGCMHypersurfaces2023},
  \begin{align*}
    \| (\nu^{\S})^l h \|_{\mathfrak{h}_{s-l}({\S})} \les r \sup_{\RR}\Big(|\widetilde{\dk}^{\leq s} h|+\dg |\dk^{\leq s} h|\Big)+ \| (\nab^{\S}_{\nu^{\S}})^{\leq l-1}(F, b^{\S}-b)\|_{\mathfrak{h}_{s-l}({\S})}\sup_{\RR}|\dk^{\leq s} h|,
  \end{align*}
  where $\widetilde{\dk}\vcentcolon= (e_3-(z +\underline{\Omega})e_4, \slashed{\dk})$ and $b\vcentcolon=-z-\underline{\Omega}$. Applying the above to $\dot{\trch}$, $\dot{\trchb}$ and $\dot{\mu}$, using condition \eqref{eq:control-trch-trchb-mu} and their $\nab_4$ equations \eqref{eq:assumption-Einstein}, we infer for $0 \leq l \leq s_{\max}+1$
  \begin{align*}
    \| (\nu^{\S})^l (\dot{\trch}, \dot{\trchb}, r\dot{\mu}) \|_{\mathfrak{h}_{s_{\max}+1-l}({\S})}
    \les{}& r^{-1} \dg + r^{-1} \epg\Big| \overline{b^{\S}}+1+\frac{2M^{\S}}{r^{\S}}-\frac{(Q^{\S})^2}{(r^{\S})^2}\Big|\\
          &+ r^{-1} \mathring{\epsilon} \sum_{j=1}^l\|(\nab^{\S}_{\nu^{\S}} )^j F \|_{\mathfrak{h}_{s_{\max}+1-j}({\S})}.
  \end{align*}
  Finally, in view of \eqref{eq:ThmGCMS1}, for every ${\S}\subset \Sigma$ we have 
  \begin{align*}
    \| F\|_{\mathfrak{h}_{s_{\max}+1}({\S})}\les \dg.
  \end{align*}
  To derive the remaining tangential derivatives of $F$ along $\Sigma$, we commute the GCM system \eqref{eq:qequivalentGCMsystemwhicisnotsolvabledirectly:2:bis} where we impose the GCM conditions \eqref{def:GCMC:00}, and use \eqref{eq:transport-J}, i.e. $\nu^{\S}( J^{({\S}, p)}) =0$. Applying Proposition \ref{prop-a-priori-estimates-f-fb}, we infer 
  \begin{align*}
    \| \nab^{\S}_{\nu^{\S}} (f, \fb), \nu(\ovla) - \ov{\nu(\ovla)}^{\S}\|_{\mathfrak{h}_{s_{\max}}({\S})}
    \les{}& \dg + |(\div^{\S} \nab^{\S}_{\nu^{\S}} \fb)_{\ell=1, J^{(\S)}}|\\
          &+(\epg + r^{-1}) \big(\| \nab^{\S}_{\nu^{\S}}F\|_{\mathfrak{h}_{s_{\max}}({\S})}+r^{-1} \| \nu \mathring{b}\|_{\mathfrak{h}_{s_{\max}}({\S})}\big).
  \end{align*}
  To estimate the $\ell=1$ modes of $\nab^{\S}_{\nu^{\S}} \fb$, we make use of \eqref{eq:nabla-nu-F} to show that 
  \begin{align*}
    \int_{\S} \div^{\S} \nab^{\S}_{\nu^{\S}}( \fb ) J^{(\S, p)}= 2(\div^{\S} \xib^{\S})_{\ell=1, J^{(\S)}}+O(\dg).
  \end{align*}
  By combining the above and choosing $\epg$ small enough and $\mathring{r}$ large enough, we deduce \eqref{eq:control-F}.
  For more details see the proof of Proposition 4.22 in \cite{shenConstructionGCMHypersurfaces2023}.
\end{proof}

\subsubsection{Step 4: Derive and solve the ODE system for \texorpdfstring{$\Psi$}{}, \texorpdfstring{$\underline{\Lambda}$}{}}

Recall from \eqref{eq:definition-Sigma} that we have 
\begin{align}
  \Sigma = \bigcup_{s \in \mathring{I}} {\S}(s), \qquad {\S}(s)\vcentcolon={\S}[\Psi(s),s, \underline{\Lambda}(s)].
\end{align}
We denote
\begin{align*}
  r(s)\vcentcolon= r^{{\S}(s)}, \qquad \underline{B}(s)\vcentcolon= (\div^{{\S}(s)} \xib^{{\S}(s)})_{\ell=1, J^{({\S}(s))}}, \qquad D(s)\vcentcolon=\overline{b^{{\S}(s)}}+1+\frac{2M_{(0)}}{r^{{\S}(s)}}-\frac{(Q_{(0)})^2}{(r^{{\S}(s)})^2}.
\end{align*}

We prove the following proposition.

\begin{proposition}
  We have the following equations for the functions $D(s)$, $r(s)$ and the triplets $\underline{\Lambda}, \underline{B}$:
  \begin{align}
    \frac{1}{\psi'(s)-1}\underline{\Lambda}'(s) &= \underline{B}(s) + \underline{G}(\underline{\Lambda}, \psi)(s) + \underline{N}(\underline{B}, D, \underline{\Lambda}, \psi)(s), \label{eq:final-ODE-Lambda}\\
    \psi'(s)&= -\frac 1 2 D(s)+H(\underline{B}, \underline{\Lambda}, \psi)(s) + M(\underline{B}, D, \underline{\Lambda}, \psi)(s) \label{eq:final-ODE-psi},
  \end{align}
  where $\underline{G}$ is a $O(1)$-Lipschitz function of $(\underline{\Lambda}, \psi)$, $H$ is a $O(1)$-Lipschitz function of $(\underline{B}, \underline{\Lambda}, \psi)$ and $\underline{N}, M$ are $O(\epg^{\frac 1 2 })$-Lipschitz functions of $(\underline{B}, D, \underline{\Lambda}, \psi)$.
\end{proposition}

\begin{proof}
  We start by computing the derivative with respect to $s$ of $\underline{\Lambda}$, which is given (see Lemma 4.33 in \cite{shenConstructionGCMHypersurfaces2023}) for a function $h$ defined on the curve of the South poles $\gamma(s)\vcentcolon=(\Psi(s), s, 0, 0)$ by 
  \begin{align*}
    X \big|_{\gamma(s)}(h)=C(s) \nu^{\S}\big|_{\gamma(s)}(h),
  \end{align*}
  where $C(s)=\frac{\lambda}{z}\Psi'(s)\big|_{\gamma(s)}+ F \c \Gab + O(F^2)$.

  We recall the following identity, see Lemma 4.34 in \cite{shenConstructionGCMHypersurfaces2023}, 
  \begin{align}\label{eq:nu-Lambda}
    \nu^{{\S}(s)}\big|_{\gamma(s)} \underline{\Lambda}(s) = \int_{{\S}(s)}\nu^{{\S}(s)}(\div^{{\S}(s)} \fb) J^{({\S}(s), p)}-\frac{4}{r(s)}\underline{\Lambda}(s)+\underline{E}(s),
  \end{align}
  where
  \begin{align*}
    \underline{E}(s)={}&\int_{{\S}(s)} \big( \frac{2}{r(s)}\widecheck{b}^{{\S}(s)}+\trchbc^{{\S}(s)} \big)(\div^{{\S}(s)} \fb) J^{({\S}(s), p)}\\
                    &+z^{{\S}(s)}\big|_{\gamma(s)} \int_{{\S}(s)}\Bigg[ \Big((z^{{\S}(s)})^{-1}- (z^{{\S}(s)})^{-1}|_{\gamma(s)}\Big) \times\\
                    &\Big(\nu^{{\S}(s)}(\div^{{\S}(s)} \fb)-\frac{4}{r(s)} \div^{{\S}(s)} \fb+\big( \frac{2}{r(s)}\widecheck{b}^{{\S}(s)}+\trchbc^{{\S}(s)} \big)(\div^{{\S}(s)} \fb)\Big)J^{({\S}(s), p)}\Bigg]\\
                    ={}& \mbox{Good},
  \end{align*}
  where the last relation is a consequence of Lemma 4.36 in \cite{shenConstructionGCMHypersurfaces2023}. Here, we denote $h=\mbox{Good}$ if $h=O(\epg)$ and $\partial_v h =O(\epg^{\frac 1 2})$ for any $h \in \{\underline{B}, D, \underline{\Lambda}, \psi \}$.

  To compute $\int_{{\S}(s)} \nu^{\S}(\div^{\S}\underline{f})J^{({\S}, p)}$ in the above,
  we start with the second identity in \eqref{eq:nabla-nu-F} and deduce
  \begin{align*}
    \int_{{\S}(s)} \nu^{\S}(\div^{\S}\underline{f})J^{({\S}, p)}
    ={}&2\underline{B}(s)-2\int_{{\S}(s)}(\div^{\S}\xib) J^{({\S}, p)}+\frac{4}{r}\underline{\Lambda}(s)\\
       &-2\int_{{\S}(s)} (\Delta^{\S}\lambda) J^{({\S}, p)} +O(r^{-2})(\Lambda(s) + \underline{\Lambda}(s)) +\mbox{Good}.
  \end{align*}
  Using that $\int_{{\S}(s)}(\div^{\S}\xib) J^{({\S}, p)}=\mbox{Good}$ and recalling that $\Lambda(s)=\int_{{\S}(s)}\div^{\S}f J^{({\S}, p)}$ is uniquely determined by Proposition \ref{prop:C-transform}
  for every $\underline{\Lambda}$ satisfying the assumptions, and satisfies \eqref{eq:estimate-Lambda}, we deduce 
  \begin{align*}
    \int_{{\S}(s)} \nu^{\S}(\div^{\S}\underline{f})J^{({\S}, p)}&=2\underline{B}(s)+\frac{4}{r}\underline{\Lambda}(s)-2\int_{{\S}(s)} \Delta^{\S}\ovla J^{({\S}, p)} +O(r^{-2})\underline{\Lambda}(s) +\mbox{Good}.
  \end{align*}
  Now, from the transformation formula of $\trch$ in Lemma \ref{lem:EM-transform}, we deduce
  \begin{align*}
    \int_{{\S}(s)}(\Delta^{\S}\ovla )J^{({\S}, p)}={}& \frac{1}{r^{\S}} \int_{{\S}(s)}\Delta^{\S}(r- r^{\S} )J^{({\S}, p)}-\frac{r^{\S}}{2} \int_{{\S}(s)}\Delta^{\S}\trchc J^{({\S}, p)}\\
                                                  &- \frac{r^{\S}}{2} \int_{{\S}(s)}\Delta^{\S}\div^{\S}(f) J^{({\S}, p)}+ \mbox{Good},
  \end{align*}
  from which we deduce, using again \eqref{eq:estimate-Lambda}, 
  \begin{align*}
    \int_{{\S}(s)} \Delta^{\S}\ovla J^{({\S}, p)}&=\frac{1}{2r}\underline{\Lambda}(s) +\mbox{Good}.
  \end{align*}
  We therefore obtain 
  \begin{align}\label{eq:nu-divifb}
    \int_{{\S}(s)} \nu^{\S}(\div^{\S}\underline{f})J^{({\S}, p)}&=2\underline{B}(s)+\frac{3}{r}\underline{\Lambda}(s) +O(r^{-2}) \underline{\Lambda}(s) +\mbox{Good}.
  \end{align}
  Plugging \eqref{eq:nu-divifb} in \eqref{eq:nu-Lambda}
  we obtain
  \begin{align*}
    \frac{1}{\Psi'(s)}\underline{\Lambda}'(s)=\frac{C(s)}{\Psi'(s)}\left(2\underline{B}(s)-r^{-1}\underline{\Lambda} +O(r^{-2})\underline{\Lambda}+\mbox{Good}\right).
  \end{align*}
  Using that $\frac{C(s)}{\Psi'(s)}=\frac 1 2 +O(\mathring{\epsilon})$ and that $\Psi(s)=-s+\psi(s)+c_0$, we obtain
  \begin{align*}
    \frac{1}{-1+\psi'(s)}\underline{\Lambda}'(s)=\underline{B}(s)+G(\underline{\Lambda}, \psi)(s)+ \underline{N}(\underline{B}, D, \underline{\Lambda}, \psi)(s),
  \end{align*}
  where 
  \begin{align*}
    G(\underline{\Lambda}, \psi)(s)=-\frac 1 2 r^{-1}\underline{\Lambda} +O(r^{-2})\underline{\Lambda}, \qquad \underline{N}=\mbox{Good}.
  \end{align*}

  Finally, the equation for $\psi$ is derived in the same way as in Proposition 4.30 in \cite{shenConstructionGCMHypersurfaces2023}.
\end{proof}

To obtain the GCM hypersurface $\Sigma_0$ and conclude the proof of Theorem \ref{thm:GCMH} by solving the closed system of equations \eqref{eq:final-ODE-Lambda}-\eqref{eq:final-ODE-psi} with initial conditions
\begin{align}
  \psi(\sg)=0, \qquad \underline{\Lambda}(\sg)=\underline{\Lambda}_0,
\end{align}
and the choice $\underline{B}=D=0$, i.e. 
\begin{align}
  \frac{1}{\breve{\psi}'(s)-1}\breve{\underline{\Lambda}}'(s) &= \underline{G}(\breve{\underline{\Lambda}}, \breve{\psi})(s) + \breve{\underline{N}}(\breve{\underline{\Lambda}}, \breve{\psi})(s), \label{eq:final-ODE-Lambda-hat}\\
  \breve{\psi}'(s)&=\breve{H}(\breve{\underline{\Lambda}}, \breve{\psi})(s) + \breve{M}(\breve{\underline{\Lambda}}, \breve{\psi})(s) \label{eq:final-ODE-psi-hat}, \\
  \breve{\psi}(\sg)&=0, \\
  \breve{\underline{\Lambda}}(\sg)&=\underline{\Lambda}_0, 
\end{align}
where $\breve{\underline{N}}(\breve{\underline{\Lambda}}, \breve{\psi})(s)\vcentcolon=\underline{N}(0, 0, \underline{\Lambda}, \psi)(s)$, $\breve{H}(\breve{\underline{\Lambda}}, \breve{\psi})(s)\vcentcolon=H(0, \underline{\Lambda}, \psi)(s)$, $\breve{M}(\breve{\underline{\Lambda}}, \breve{\psi})(s)\vcentcolon=M(0, 0, \underline{\Lambda}, \psi)(s)$.

Applying the Cauchy-Lipschitz theorem, we deduce that the system admits a unique solution $(\breve{\underline{\Lambda}}, \breve{\psi})$ defined in a small neighborhood $\mathring{I}$ of $\sg$ which satisfies $|\breve{\underline{\Lambda}}(s), \breve{\psi}(s)|\les \dg$.

The desired hypersurface $\Sigma_0$ is given from this solution $(\breve{\underline{\Lambda}}(s), \breve{\psi}(s))$ by 
\begin{align}
  \Sigma_0 \vcentcolon= \bigcup_{s \in \mathring{I}}{\S}[\breve{\Psi}(s),s, \breve{\underline{\Lambda}}(s)].
\end{align}
By subtracting the system \eqref{eq:final-ODE-Lambda-hat}-\eqref{eq:final-ODE-psi-hat} from the one \eqref{eq:final-ODE-Lambda}-\eqref{eq:final-ODE-psi} applied to $(\breve{\underline{\Lambda}}, \breve{\psi})$ we deduce that $\underline{B}=D=0$ on $\Sigma_0$.

Finally, applying Proposition \ref{prop:estimates-frame-coefficients} to $\Sigma_0$ we deduce 
\begin{align}
  \| F\|_{\mathfrak{h}_{s_{\max}+1}({\S})} + \| \nab^{\S}_{\nu^{\S}} F\|_{\mathfrak{h}_{s_{\max}}({\S})} \les \mathring{\de}.
\end{align}
By construction, the statements 1--5 of Theorem \ref{thm:GCMH} are satisfied by $\Sigma_0$.
Using the transversality conditions \eqref{eq:transversality-condition} and the transformation formulas for $\xi$, $\ze$, $\etab$, $\om$ in Lemma \ref{lem:EM-transform}, we finally deduce \eqref{eq:bounds-f-fb-la}. 

\printbibliography

\end{document}
